\documentclass[11pt]{article}
\usepackage{authblk}
\title{Ergodic Theorems for PSPACE functions and their converses}
\author[1]{Satyadev Nandakumar}
\author[1]{Subin Pulari}
\affil[1]{
  Department of Computer Science and Engineering\\
  Indian Institute of Technology Kanpur,
  Kanpur, Uttar Pradesh, India.
}
\affil[]{\{\textit{satyadev,subinp}\}@cse.iitk.ac.in}

\usepackage{amsmath,amsthm,amssymb}
\usepackage[margin=1in]{geometry}
\usepackage[colorlinks=true]{hyperref}

\usepackage{amsfonts}
\usepackage{amsmath} 
\usepackage{amssymb}
\usepackage{enumerate}

\newcommand{\PROOF}{\begin{proof}}

\newcommand{\QED}{\end{proof}}

\newcommand{\N}{\mathbb{N}}

\newcommand{\Q}{\mathbb{Q}}
\newcommand{\R}{\mathbb{R}}

\newcommand{\SUBEXP}{\mathrm{SUBEXP}}

\newcommand{\PSPACE}{{\rm PSPACE}}

\newcommand{\EXP}{{\rm EXP}}

\usepackage[svgnames]{xcolor}

\def\<{\left\langle}
\def\>{\right\rangle}

\theoremstyle{plain}
\newtheorem{theorem}{Theorem}[section]
\newtheorem{lemma}[theorem]{Lemma}
\newtheorem{corollary}[theorem]{Corollary}
\newtheorem*{theorem*}{Theorem}

\theoremstyle{definition}

\newtheorem{definition}[theorem]{Definition}

\usepackage{multirow}
\usepackage{array}
\usepackage{tikz-cd}
\usepackage{amsfonts}

\begin{document}

\maketitle
\begin{abstract}
We initiate the study of effective pointwise ergodic theorems in
resource-bounded settings. Classically, the convergence of the ergodic
averages for integrable functions can be arbitrarily slow
\cite{Kre78}. In contrast, we show that for a class of $\PSPACE$ $L^1$
functions, and a class of $\PSPACE$ computable measure-preserving
ergodic transformations, the ergodic average exists for all $\PSPACE$
randoms and is equal to the space average on every $\EXP$ random. We
establish a partial converse that $\PSPACE$ non-randomness can be
characterized as non-convergence of ergodic averages. Further, we
prove that there is a class of resource-bounded randoms, \emph{viz.}
$\SUBEXP$-space randoms, on which the corresponding ergodic theorem
has an exact converse - a point $x$ is $\SUBEXP$-space random if and
only if the corresponding effective ergodic theorem holds for $x$.
\end{abstract}

\section{Introduction}
In Kolmogorov's program to found information theory on the theory of
algorithms, we investigate whether individual ``random'' objects obey
probabilistic laws, \emph{i.e.}, properties which hold in sample
spaces with probability 1. Indeed, a vast and growing literature
establishes that \emph{every} Martin-L\"of random sequence (see for
example, \cite{downey} or \cite{Nies2009}) obeys the Strong Law of
Large Numbers \cite{vanL87a}, the Law of Iterated Logarithm
\cite{Vovk88}, and surprisingly, the Birkhoff Ergodic Theorem
\cite{Vyug97}, \cite{Nan08}, \cite{effprobtomltheory}, \cite{BDHMS10}
and the Shannon-McMillan-Breiman theorem \cite{Hochman09},
\cite{Hoy12}, \cite{RuteThesis}. In effective settings, the theorem
for Martin-L\"of random points implies the classical theorem since the
set of Martin-L\"of randoms has Lebesgue measure 1, and hence is
stronger. 

In this work, we initiate the study of ergodic theorems in
resource-bounded settings. This is a difficult problem, since
classically, the convergence speed in ergodic theorems is known to be
arbitrarily slow (\emph{e.g.} see Bishop \cite{Bish67}, Krengel
\cite{Kre78}, and V'yugin \cite{Vyug97}). However, we establish
ergodic theorems in resource-bounded settings which hold on every
resource-bounded random object of a particular class.  The main
technical hurdle we face is the lack of sharp tail bounds. The only
general tail bound in ergodic settings is the maximal ergodic
inequality, which yields only an inverse linear bound in the number of
sample points, in contrast to the inverse exponential bounds in the
Chernoff and the Azuma-Hoeffding inequalities.

Rapid $L^1$ convergence of subsequences of ergodic averages suffices
to establish that the ergodic average at all $\PSPACE$ randoms exist,
and is equal to the space average on all $\EXP$ randoms. A
non-trivial connection with the theory of uniform distribution of
sequences modulo 1 \cite{kuipersniederreiter},
\cite{maxfield1952short}, \cite{Pillai40},
\cite{nandakumar2019analogue} enables us to show that the canonical
example of the Bernoulli measure and the left shift satisfies our
convergence assumption. In general, such assumptions are unavoidable
since an adaptation of V'yugin's counterexample \cite{Vyug97} shows
that there are $\PSPACE$ computable ergodic Markov systems where the
convergence rate to the ergodic average is not even computable.

It is known that at the level of Martin-L\"of randomness, we can use
ergodicity to characterize randomness. Franklin and Towsner
\cite{franklin2014randomness} show that for every non-Martin-L\"of
random $x$, there is an effective ergodic system where the ergodic
average at $x$ does not converge to the space average. We first show
that our $\PSPACE$ effective ergodic theorem admits a partial converse
of this form. $\PSPACE$ non-randoms can be characterized as points
where the $\PSPACE$ ergodic theorem fails. Since the theorem holds on
the smaller set of $\EXP$ randoms, it is important to know
whether there is a class of resource-bounded randoms on which an
effective ergodic theorem holds with an exact converse. We show that
the class of $\SUBEXP$-space randoms is one such - on every
$\SUBEXP$-space randoms, the $\SUBEXP$-space ergodic theorem holds,
and on every $\SUBEXP$-space non-random, it fails. We summarize our
results in Table \ref{tbl}.

The proofs of these results are adapted from the techniques of Rute
\cite{RuteThesis}, Ko \cite{Ko2012}, Galatolo, Hoyrup \& Rojas
\cite{hoyruprojasgalatolo}, \cite{hoyrup2009}, and Huang \& Stull
\cite{HuaStu16}.\footnote{There are alternative approaches to the
proof in Martin-L\"of settings, like that of V'yugin
\cite{Vyug97}. However, the tool he uses for establishing the result
is a lower semicomputable test defined on infinite sequences - this is
difficult to adapt to resource-bounded settings requiring the output
value within bounded time or space. Moreover, the functions in
V'yugin's approach are continuous. We consider the larger class of
$L^1$ functions, which can be discontinuous in general. } Our proofs
involve several new quantitative estimates, which may of general
interest.

\begin{table}[!h]
\centering
\begin{tabular}{|>{\centering\arraybackslash} m{3cm}| >{\centering\arraybackslash} m{4cm}| >{\centering\arraybackslash} m{4cm}|}
\hline
\multirow{2}{*}{Class of functions} & \multicolumn{2}{c|}{Convergence of ergodic averages} \\ \cline{2-3}
& $\forall f (A^f_n \to \int f d\mu)$ & $\exists f (A^f_n \not\to \int f d\mu)$ \\
\hline 
$\PSPACE$ $L^1$ & $\EXP$ randoms (Theorem \ref{thm:pspaceergodictheorem}) & $\PSPACE$ non-randoms  (Theorem \ref{thm:conversetopspaceergodictheorem})\\
\hline
$\SUBEXP$-space $L^1$ & $\SUBEXP$-space randoms (Theorem \ref{thm:subexpergodictheorem}) & $\SUBEXP$-space non-randoms  (Theorem \ref{thm:conversetosubexpergodictheorem})\\
\hline
\end{tabular}
\caption{Summary of the results involving $\PSPACE$/$\SUBEXP$-space
  systems.}
\label{tbl}
\end{table}

\section{Preliminaries}
\label{sec:preliminaries}
Let $\Sigma = \{0,1\}$ be the binary alphabet. Denote the set of all
finite binary strings by $\Sigma^*$ and the set of infinite binary
strings by $\Sigma^\infty$. For $\sigma \in \Sigma^*$ and $y \in
\Sigma^* \cup \Sigma^\infty$, we write $\sigma \sqsubseteq y$ if
$\sigma$ is a prefix of $y$. For any infinite string $y$ and any
finite string $\sigma$, $\sigma[n]$ and $y[n]$ denotes the character
at the $n$\textsuperscript{th} position in $y$ and $\sigma$
respectively. For any infinite string $y$ and any finite string
$\sigma$, $\sigma[n,m]$ and $y[n,m]$ represents the strings
$\sigma[n]\sigma[n+1]\dots \sigma[m]$ and $y[n]y[n+1]\dots y[m]$
respectively. For any $x \in \Sigma^{\infty}$ and $n \in \N$, $x \upharpoonleft n$ denotes the string $x[1,n]$. We denote finite strings using small Greek letters like
$\sigma$, $\alpha$ etc. The length of a finite binary string $\sigma$
is denoted by $|\sigma|$.

For any finite string $\sigma$, the \emph{cylinder} $[\sigma]$ is the
set of all infinite sequences with $\sigma$ as a prefix.
$\chi_{\sigma}$ denotes the characteristic function of $[\sigma]$. For
any set of strings $S \subseteq \Sigma^*$, $[S]$ is the union of
$[\sigma]$ over all $\sigma \in S$. Extending the notation, $\chi_{S}$
denotes the characteristic function of $[S]$. The Borel
$\sigma$-algebra generated by the set of all cylinders is denoted by
$\mathcal{B}(\Sigma^\infty)$.

Unless specified otherwise, any $n\in \N$ is represented in the binary
alphabet. As is typical in resource-bounded settings, some integer
parameters are represented in unary. The set of unary strings is
represented as $1^*$, and the representation of $n \in \N$ in unary is
$1^n$, a string consisting of $n$ ones. For any $n_1,n_2 \in \N$,
$[n_1,n_2]$ represents the set $\{n \in \N: n_1 \leq n \leq n_2\}$.

Throughout the paper we take into account the number of cells used in
the output tape and the working tape when calculating the space
complexity of functions. We assume a finite representation for the set
of rational numbers $\Q$ satisfying the following: there exists a $c
\in \N$ such that if $r \in \Q$ has a representation of length $l$
then $r \leq 2^{l^c}$ . Following the works of Hoyrup, and Rojas
\cite{hoyrup2009}, we introduce the notion of a $\PSPACE$-computable
probability space on the Cantor space by endowing it with a
$\PSPACE$-computable probability measure.

\begin{definition}
  Consider the probability space $(\Sigma^\infty,
  \mathcal{B}(\Sigma^\infty))$.  A Borel probability measure $\mu:
  \mathcal{B}(\Sigma^\infty) \to [0,1]$, is a
  \emph{$\PSPACE$-probability measure} if there is a
  $\PSPACE$ machine $M: \Sigma^* \times 1^* \to \Q$ such
  that for every $\sigma \in \Sigma^*$, and $n \in \N$, we have that
  $|M(\sigma, 1^n) - \mu([\sigma])| \le 2^{-n}$. 
  
  
  A \emph{$\PSPACE$-probability Cantor space} is a pair $(\Sigma^\infty,\mu)$
  where $\Sigma^\infty$ is the Cantor space, and $\mu$ is a $\PSPACE$ probability
  measure.
\end{definition}

In order to define $\PSPACE$ ($\EXP$) randomness using $\PSPACE$ ($\EXP$) tests we require the following method for approximating sequences of open sets in $\Sigma^\infty$ in polynomial space (exponential time).

\begin{definition}[$\PSPACE$/$\EXP$ sequence of open sets
    \cite{HuaStu16}]  
  \label{def:pspaceopensets}
  A sequence of open sets $\<U_n\>_{n=1}^\infty$ is
  a \emph{$\PSPACE$ sequence of open sets} if there exists a sequence
  of sets $\<S^k_n\>_{k,n \in \N}$, where $S^k_n \subseteq \Sigma^*$
  such that 
  \begin{enumerate}
  \item $U_n = \cup_{k=1}^{\infty}[S^{k}_n]$, where for any
    $m>0$, $\mu\left(U_n-\cup_{k=1}^m [S^k_n]\right)\leq
    \frac{1}{2^m}$.
  \item There exists a \emph{controlling polynomial} $p$ such
    that $max\{|\sigma|:\sigma \in \cup_{k=1}^m S^k_n)\}\leq p(n+m)$.
  \item The function $g:\Sigma^* \times 1^*\times 1^* \to \{0,1\}$
    such that $g(\sigma,1^n,1^m)=1$ if $\sigma \in S^m_n$, and 0
    otherwise, 
    is decidable by a $\PSPACE$ machine.
  \end{enumerate}
  The definition of \emph{$\EXP$ sequence of open sets} is similar but
  the bound in condition 2 is replaced with $2^{p(n+m)}$ and the
  machine in condition 3 is an $\EXP$-time machine.
\end{definition}


Henceforth, we study the notion of resource bounded randomness on $(\Sigma^\infty, \mu)$. 

\begin{definition}[$\PSPACE$/$\EXP$ randomness \cite{stull}] 
\label{def:pspacetest}
  A sequence of open sets $\<U_n\>_{n=1}^\infty$ is a \emph{$\PSPACE$
    test} if it is a $\PSPACE$ sequence of open sets and for all $n
  \in \N$, $\mu(U_n)\leq \frac{1}{2^n}$.
  
  A set $A \subseteq \Sigma^\infty$ is \emph{$\PSPACE$ null} if there
  is a $\PSPACE$ test $\<U_n\>_{n=1}^\infty$ such that $A \subseteq
  \cap_{n=1}^{\infty} U_n$. A set $A\subseteq \Sigma^\infty$
  is \emph{$\PSPACE$ random} if $A$ is not $\PSPACE$ null.
  
  The $\EXP$ analogues of the above concepts are defined similarly
  except that $\<U_n\>_{n=1}^\infty$ is an $\EXP$ sequence of open
  sets.
\end{definition}

By considering the sequence $\<\cup_{i=1}^{k} S^i_n\>_{k,n
  \in \N}$ instead of $\<S^k_n\>_{k,n \in \N}$, without loss of
generality, we can assume that for each $n$, $\<S^k_n\>_{k=1}^{\infty}$
is an increasing sequence of sets.

In order to establish our ergodic theorem, it is convenient to define
a $\PSPACE$ version of Solovay tests, where the relaxation is that the
measures of the sets $U_n$ can be any sufficiently fast convergent
sequence. We later show that this captures the same set of randoms as
$\PSPACE$ tests.

\begin{definition}[$\PSPACE$ Solovay test]
  A sequence of open sets $\<U_n\>_{n=1}^\infty$ is
  a \emph{$\PSPACE$ Solovay test} if it is a $\PSPACE$ sequence of
  open sets and there is a polynomial $p$ such that $\forall m \geq 0$,
          $\sum_{n=p(m)+1}^{\infty}\mu(U_n) \leq
          \frac{1}{2^m}$\footnote{This implies that $\sum_{n=1}^{\infty}\mu(U_n) < \infty $}. A set $A \subseteq \Sigma^\infty$ is \emph{$\PSPACE$ Solovay null}
  if there exists a $\PSPACE$ Solovay test $\<U_n\>_{n=1}^\infty$ such
  that $A \subseteq
  \cap_{i=1}^{\infty}\cup_{n=i}^{\infty}
  U_n$. $A\subseteq \Sigma^\infty$ is \emph{$\PSPACE$ Solovay random}
  if $A$ is not $\PSPACE$ Solovay null.  
\end{definition}

\begin{theorem}
  \label{thm:pspacesolovayequivalence}
A set $A \subseteq \Sigma^\infty$ is $\PSPACE$ null if and only if $A$
is $\PSPACE$ Solovay null. 	 
\end{theorem}
\begin{proof}
It is easy to see that if $A$ is $\PSPACE$ null then $A$ is $\PSPACE$
Solovay null. Conversely, let $A$ be $\PSPACE$ Solovay null and let
$\<U_n\>_{n=1}^{\infty}$ be any Solovay test which witnesses this
fact. Let $V_{n}=\cup_{i=p(n)+1}^{\infty} U_n$. We show that
$\<V_n\>_{n=1}^\infty$ is a $\PSPACE$ test. Let $\<S_n^k\>_{n,k \in
  \N}$ be any sequence of sets approximating $\<U_n\>_{n=1}^{\infty}$
as in definition \ref{def:pspaceopensets} such that
$\<S_n^k\>_{k=1}^{\infty}$ is increasing for each $n$. We define
a sequence of sets $\<T_n^k\>_{n,k \in \N}$ approximating $V_n$ as
follows.

Let $r(n,k)=\max \{p(n)+1,p(k+1)\}$. Define
\begin{align*}
T_n^k = \bigcup\limits_{i=p(n)+1}^{r(n,k)} S_i^{r(n,k)-p(n)+k+1}.	
\end{align*}

We can easily verify the first three conditions in definition
\ref{def:pspaceopensets}. Using the machine $M$ and controlling polynomial $p$ witnessing that
$\<U_n\>_{n=1}^{\infty}$ is a $\PSPACE$ sequence of open sets, we can
construct the corresponding machines for $\<V_n\>_{n=1}^{\infty}$ in the following way. Machine $N$ on input $(\sigma,1^n,1^k)$ does the following:
\begin{enumerate}
	\item For each $i \in [p(n)+1,r(n,k)]$ do the following:
	\begin{enumerate}
	\item Output $1$ if $M(\sigma,1^i,1^{r(n,k)-p(n)+k+1})=1$. 
	\end{enumerate}
	\item Output $0$ if none of the above computations results in $1$.
\end{enumerate}
It is straightforward to verify that $N$ is a $\PSPACE$ machine.
\end{proof}

The set of $\PSPACE$ Solovay randoms and $\PSPACE$ randoms are equal,
hence to prove $\PSPACE$ randomness results, it suffices to form
Solovay tests.

\section{$\PSPACE$ $L^1$ computability}
\label{sec:pspacel1computability}

The resource-bounded ergodic theorems in our work hold for
$\PSPACE$-$L^1$ functions, the $\PSPACE$ analogue of integrable
functions. In this section, we briefly recall standard definitions for
$\PSPACE$ computable $L^1$ functions and
measure-preserving transformations. The justifications and proofs of
equivalences of various notions are present in Stull's thesis
\cite{stull2017algorithmic} and \cite{stull}.  We initially define
$\PSPACE$ sequence of simple functions, and define $\PSPACE$ integrable
functions based on approximations using these functions.

\begin{definition}[$\PSPACE$ sequence of simple functions \cite{stull}]
\label{def:pspacesequenceofsimplefunctions}
A sequence of simple functions $\<f_n\>_{n=1}^{\infty}$ where
each $f_n : \Sigma^\infty \to \Q$ is a \emph{$\PSPACE$
  sequence of simple functions} if
\begin{enumerate}
\item There is a \emph{controlling polynomial} $p$ such that for each
  $n$, there exist $k(n)\in \N$, $\{d_1,d_2, \dots,
  d_{k(n)}\}\subseteq \Q$ and $\{\sigma_1,\sigma_2, \dots,
  \sigma_{k(n)}\} \subseteq \Sigma^{p(n)}$ satisfying $f_n =
  \sum_{i=1}^{k(n)}d_i \chi_{\sigma_i}$.
\item There is a $\PSPACE$ machine $M$ such that for each $n \in \N,
  \sigma \in \Sigma^{*}$
\begin{align*}
M(1^n,\sigma)=\begin{cases}
f_n(\sigma0^\infty) & \text{if } |\sigma| \geq p(n)\\
? & \text{otherwise.}
\end{cases}
\end{align*}
\end{enumerate}
\end{definition}

Note that since $M$ is a $\PSPACE$ machine, $\{d_1,d_2 \dots
d_{k(n)}\}$ is a set of $\PSPACE$ representable numbers. Now, we
define $\PSPACE$ $L^1$-computable functions in terms of limits of
convergent $\PSPACE$ sequence of simple functions.

\begin{definition}[$\PSPACE$ $L^1$-computable functions \cite{stull}]
\label{def:pspacel1computablefunction}
A function $f \in L^1(\Sigma^\infty,\mu)$ is $\PSPACE$
$L^1$-computable if there exists a $\PSPACE$ sequence of simple
functions $\<f_n\>_{n=1}^{\infty}$ such that for every $n \in \N$,
$\lVert f-f_n \rVert \leq 2^{-n}$. The sequence
$\<f_n\>_{n=1}^{\infty}$ is called a \emph{$\PSPACE$
$L^1$-approximation} of $f$.
\end{definition}

A sequence of $L^1$ functions $\<f_n\>_{n=1}^{\infty}$ converging to
$f$ in the $L^1$-norm need not have pointwise limits. Hence the
following concept (\cite{RuteThesis}) is important in studying the
pointwise ergodic theorem in the setting of $L^1$-computability
\begin{definition}[$\widetilde{f}$ for $\PSPACE$ $L^1$-computable $f$]
Let $f \in L^1(\Sigma^\infty,\mu)$ be $\PSPACE$ $L^1$-computable and let
$\<f_n\>_{n=1}^{\infty}$ be any $\PSPACE$ sequence of simple functions
in $L^1(\Sigma^\infty,\mu)$ approximating $f$ (as in Definition
\ref{def:pspacel1computablefunction}). Define
$\widetilde{f}:\Sigma^\infty \to \R \cup \{\text{undefined}\}$ by
$\widetilde{f}(x) = \lim_{n \to \infty} f_n(x)$ if this limit exists,
and is undefined otherwise.\footnote{The definition of $\widetilde{f}$ is dependent on the choice of the approximating sequence $\<f_n\>_{n=1}^{\infty}$. However, due to Lemma \ref{lem:convergencelemma}, we use $\widetilde{f}$ in a sequence independent manner.}
\end{definition}

%
%
%

To define ergodic averages, we restrict ourselves to the following
class of transformations.
\begin{definition}[$\PSPACE$ simple transformation]
A measurable function $T:(\Sigma^\infty,\mu) \to (\Sigma^\infty,\mu)$ is a
$\PSPACE$ simple transformation if there is a \textit{controlling constant} $c$ and a
$\PSPACE$ machine $M$ such that such that for any $\sigma \in \Sigma^*$
\begin{align*}
  T^{-1}([\sigma]) = \bigcup\limits_{i=1}^{k(\sigma)} [\sigma_i]	
\end{align*}
such that:
\begin{enumerate}
\item $\{\sigma_i\}_{i=1}^{k(\sigma)}$ is a prefix free set and for all $1\leq i \leq k(\sigma)$, $|\sigma_i|\leq |\sigma|+c$
\item For each $\sigma,\alpha \in \Sigma^*$,
  \begin{align*}
    M(\sigma,\alpha) =  \begin{cases}
      1 & \text{if } |\alpha|\geq |\sigma|+c \text{ and }
      \alpha 0^\infty \in T^{-1}([\sigma])\\
      0 & \text{if } |\alpha|\geq |\sigma|+c \text{ and }
      \alpha 0^\infty \not\in T^{-1}([\sigma])\\
    ? & \text{otherwise }
  \end{cases}	
		\end{align*}
	\end{enumerate}
\end{definition}
It is easy to verify that if $T$ is a $\PSPACE$ simple transformation
then for any $n \geq 2$, $T^n$ is also a $\PSPACE$ simple
transformation. We need the following stronger assertion in the proof of the ergodic theorem.
\begin{lemma}
\label{lem:pspacetransformationcomputation}
	Let $T:(\Sigma^\infty,\mu) \to (\Sigma^\infty,\mu)$ be a $\PSPACE$ simple transformation with controlling constant $c$. There exists a $\PSPACE$ machine $N$ such that for each $n \in \N$ and $\sigma,\alpha \in \Sigma^*$,
	\begin{align*}
    N(1^n,\sigma,\alpha) =  \begin{cases}
      1 & \text{if } |\alpha|\geq |\sigma|+cn \text{ and }
      \alpha 0^\infty \in T^{-n}([\sigma])\\
      0 & \text{if } |\alpha|\geq |\sigma|+cn \text{ and }
      \alpha 0^\infty \not\in T^{-n}([\sigma])\\
    ? & \text{otherwise }
  \end{cases}	
		\end{align*}
\end{lemma}

\begin{proof}
Let $M$ be the machine witnessing the fact that $T$ is a $\PSPACE$ simple transformation with the polynomial space complexity bound $p(n)$. Let the machine $N$ do the following on input $(1^n,\sigma,\alpha)$:
\begin{enumerate}
	\item If $\alpha < \lvert \sigma \rvert+cn$, then output ?.
	\item If $n=1$ then, run $M(\sigma,\alpha)$ and output the result of this simulation.
	\item Else:
		\begin{enumerate}
		\item\label{talgorithmelsestep} For all strings $\alpha'$ of length $\lvert \sigma \rvert+c(n-1)$ do the following:
			\begin{enumerate}
			\item If $N(1^{n-1},\sigma,\alpha')=1$ then, output $1$ if $M(\alpha',\alpha)=1$. 
			\end{enumerate}
		\end{enumerate}
	\item If no output is produced in the above steps, output $0$.
\end{enumerate}
When $n=1$, $N$ uses at most $p(|\sigma|+|\alpha|+cn)+O(1)$ space. Inductively, assume that for $n=k$, $N$ uses at most $(2k-1)p(\lvert \sigma\rvert +\lvert \alpha\rvert+cn)+O(1)$ space. For $n=k+1$, the storage of $\alpha'$ and the two simulations inside step \ref{talgorithmelsestep} can be done in $2p(\lvert \sigma\rvert +\lvert \alpha\rvert+cn)+(2k-1)p(\lvert \sigma\rvert +\lvert \alpha\rvert+cn)+O(1)=(2(k+1)-1)p(\lvert \sigma\rvert +\lvert \alpha\rvert+cn)+O(1)$ space. Hence, $N$ is a $\PSPACE$ machine.
\end{proof}

$\PSPACE$ computability as defined above, relates naturally to
convergence of $L^1$ norms. But the pointwise ergodic theorem deals
with almost everywhere convergence, and its resource-bounded versions
deal with convergence on every random point.  We introduce the modes
of convergence we deal with in the present work.

\begin{definition}[$\PSPACE$-rapid limit point]
  A real number $a$ is a \emph{$\PSPACE$-rapid limit point} of the real
  number sequence $\<a_n\>_{n=1}^{\infty}$ if there exists a
  polynomial $p$ such that for all $m \in \N$, $ \exists k \leq 2^{p(m)}
  \text{ such that }|a_{k}-a| \leq \frac{1}{2^m}$. 
\end{definition}

Note that this requires rapid convergence only on a subsequence. We remark that the above is equivalent to the existence of a $\PSPACE$ machine computing the speed of convergence on input $1^m$. The
following definition is the $L^1$ version of the above.

\begin{definition}[$\PSPACE$-rapid $L^1$-limit point]
A function $f \in L^1 (\Sigma^\infty,\mu)$ is a \emph{$\PSPACE$-rapid
  $L^1$-limit point} of a sequence $\<f_n\>_{n=1}^{\infty}$ of
functions in $L^1 (\Sigma^\infty,\mu)$ if 0 is a $\PSPACE$-rapid limit point of
$\lVert f_n-f\rVert_1$.	 
\end{definition}

%

Now we define $\PSPACE$ analogue of almost everywhere
convergence (\cite{RuteThesis}).

\begin{definition}[$\PSPACE$-rapid almost everywhere convergence]
  A sequence of measurable functions $\<f_n\>_{n=1}^{\infty}$ is
  \emph{$\PSPACE$-rapid almost everywhere convergent} to a measurable
  function $f$ if there exists a polynomial $p$ such that for all
  $m_1$ and $m_2$,
  \begin{align*}
    \mu\left( \left\{ x : \sup\limits_{n \geq 2^{p(m_1+m_2)}}
    |f_n(x)-f(x)| \geq \frac{1}{2^{m_1}}\right\} \right) \leq
    \frac{1}{2^{m_2}}.
  \end{align*}
\end{definition}
{\bf Notation.} Let $A_n^{f,T}=\frac{f+f\circ T + f\circ T^2 +\dots
  f\circ T^{n-1}}{n}$ denote the $n$\textsuperscript{th} Birkhoff
average for any function $f$ and transformation $T$. We prove the
ergodic theorem in measure preserving systems where $\int f d\mu$ is a
$\PSPACE$-rapid $L^1$-limit point of $A_n^{f,T}$. In the rest of the
paper we denote $A_n^{f,T}$ simply by $A_n^f$. The
transformation $T$ involved in the Birkhoff sum is implicit.

\begin{definition}[$\PSPACE$ ergodic transformations]
A measurable function $T:(\Sigma^\infty,\mu) \to (\Sigma^\infty,\mu)$
is \emph{$\PSPACE$ ergodic} if $T$ is a $\PSPACE$ simple measure preserving
transformation such that for any $\PSPACE$ $L^1$-computable $f\in L^1
(\Sigma^\infty,\mu)$, $\int f d\mu$ is a $\PSPACE$-rapid $L^1$ limit
point of $A_n^{f}$.
\end{definition}

V'yugin \cite{Vyug97} shows that the speed of a.e. convergence to
ergodic averages in computable ergodic systems is not computable in
general. This leads us to consider some assumption on the rapidity of
convergence in resource-bounded settings.  We show that the
requirement on $L^1$ rapidity of convergence of $A^f_n$ is sufficient
to derive our result. Several probabilistic laws like the Law of Large
Numbers, Law of Iterated Logarithm satisfy this criterion, hence the
assumption is sufficiently general. Moreover, as we show now, in the
canonical example of Bernoulli systems with the left-shift, every
$\PSPACE$ $L^1$ function exhibits $\PSPACE$ rapidity of $A^f_n$, showing
that the latter property is not artificial.  The proof of this theorem
is a non-trivial application of techniques from uniform distribution
of sequences modulo 1 \cite{kuipersniederreiter}, \cite{Pillai40},
\cite{maxfield1952short}, \cite{nandakumar2019analogue}.

\begin{theorem}
\label{thm:l1functionpspacerapid}
Let $f\in L^1(\Sigma^\infty,\mathcal{B}(\Sigma^\infty),\mu)$ where
$\mu$ is the Bernoulli measure $\mu(\sigma)=\frac{1}{2^{\lvert \sigma
    \rvert}}$ and let $T$ be the left shift transformation. If $f$ is
$\PSPACE$ $L^1$-computable, then there exists a polynomial $q$
satisfying the following: given any $m\in \N$, for all $n \geq
2^{q(m)}$, $\lVert A_n^f -\int f d\mu \rVert_1 \leq 2^{-m}$.
\end{theorem}

An equivalent statement is the following: The left-shift transformation
on the Bernoulli probability measure is $\PSPACE$ ergodic\footnote{
 Equivalently, there exists a constant $c$ such
 that for all $n>0$, $\lVert A_n^f-\int f d\mu \rVert_1 \leq
 2^{-\lfloor \log(n)^\frac{1}{c} \rfloor}$.}.
Theorem \ref{thm:l1functionpspacerapid} gives an explicit bound on the
speed of convergence in the $L^1$ ergodic theorem for an interesting
class of functions over the Bernoulli space. Such bounds do not exist
in general for the $L^1$ ergodic theorem as demonstrated by Krengel in
\cite{Kre78}.

The above theorem can be obtained from the following assertion regarding $\PSPACE$-rapid convergence of characteristic functions of long enough cylinders. 

\begin{lemma}
	\label{lem:characteristicfunctionpspacerapid}
 Let $T$ be the left shift transformation $T:(\Sigma^\infty,\mathcal{B}(\Sigma^\infty),\mu) \to (\Sigma^\infty,\mathcal{B}(\Sigma^\infty),\mu)$ where $\mu$ is the Bernoulli measure $\mu(\sigma)=2^{-\lvert \sigma \rvert}$. There exist polynomials $q_1,q_2$ such that for any $m \in \N$ and $\sigma \in \Sigma^*$ with $|\sigma|\geq q_1(m)$ we get $\lVert A_n^{\chi_{\sigma}} - \mu(\sigma) \rVert_1 \leq 2^{-m}$ for all $n \geq \lvert \sigma \rvert^3 2^{q_2(m)}$.
\end{lemma}

 Now we prove Theorem \ref{thm:l1functionpspacerapid} by assuming Lemma \ref{lem:characteristicfunctionpspacerapid}.

\begin{proof}[Proof of Theorem \ref{thm:l1functionpspacerapid}]
	Let $\<f_n\>_{n=1}^{\infty}$ be a $\PSPACE$ sequence of simple functions witnessing the fact that $f$ is $\PSPACE$ $L^1$-computable. Let $p$ be a controlling polynomial and let $t$ be a polynomial upper bound for the space complexity of the machine associated with $\<f_n\>_{n=1}^{\infty}$. Let $q_1$,$q_2$ be the polynomials from Lemma \ref{lem:characteristicfunctionpspacerapid}. Let $c \in \N$ be any number such that if a $r \in \Q$ has a representation of length $l$ then $r \leq 2^{l^c}$ (see Section \ref{sec:preliminaries}). Observe that for any $m \in \N$,
	\begin{align*}
	\lVert A_n^f -\int f d\mu \rVert_1 &\leq \lVert A_n^{f}- A_n^{f_{q_{1}(m+3)}} \rVert_1  + \lVert A_n^{f_{q_{1}(m+3)}} -\int f_{q_{1}(m+3)} d\mu \rVert_1 + \lVert \int f_{q_{1}(m+3)} d\mu - \int f d\mu \rVert_1 \\
	&\leq \frac{1}{2^{q_{1}(m+3)}} + \lVert A_n^{f_{q_{1}(m+3)}} -\int f_{q_{1}(m+3)} d\mu \rVert_1 + \frac{1}{2^{q_{1}(m+3)}}.\\
	&\leq \frac{1}{2^{m+3}} + \lVert A_n^{f_{q_{1}(m+3)}} -\int f_{q_{1}(m+3)} d\mu \rVert_1 + \frac{1}{2^{m+3}}.
	\end{align*}
	We know that there exists $\{\sigma_1,\sigma_2 \dots \sigma_k\} \subseteq \Sigma^{p(q_{1}(m+3))}$ such that $A_n^{f_{q_{1}(m+3)}}=\sum_{i=1}^{k(q_{1}(m+3))} d_i \chi_{\sigma_i}$ where each $d_i \leq 2^{t(q_{1}(m+3)+p(q_{1}(m+3)))^c}$. Hence,
	\begin{align*}
		\lVert A_n^{f_{q_{1}(m+3)}} -\int f_{q_{1}(m+3)} d\mu \rVert_1 \leq 2^{t(q_{1}(m+3)+p(q_{1}(m+3)))^c}\sum\limits_{i=1}^{k(q_{1}(m+3))} \lVert A_n^{\chi_{\sigma_i}} - \mu(\sigma_i) \rVert_1
	\end{align*}
	Since $\lvert \sigma_i \rvert \geq p(q_{1}(m+3)) \geq q_{1}(m+3)$, using Lemma \ref{lem:characteristicfunctionpspacerapid}, for 
	\begin{align*}
		n \geq p(q_{1}(m+3))^3 2^{q_2(t(q_{1}(m+3)+p(q_{1}(m+3)))^c+p(q_{1}(m+3))+m+3)}
	\end{align*}
	we get that,
	\begin{align*}
		\lVert A_n^{f_{q_{1}(m+3)}} -\int f_{q_{1}(m+3)} d\mu \rVert_1 &\leq  \frac{2^{t(q_{1}(m+3)+p(q_{1}(m+3)))^c+p(q_{1}(m+3))}}{2^{t(q_{1}(m+3)+p(q_{1}(m+3)))^c+p(q_{1}(m+3))+m+3}}\\ 
		&\leq \frac{1}{2^{m+3}}.
	\end{align*}
	Hence, for all $n\geq p(q_{1}(m+3))^3 2^{q_2(t(q_{1}(m+3)+p(q_{1}(m+3)))^c+p(q_{1}(m+3))+m+3)}$ we have $\lVert A_n^f -\int f d\mu \rVert_1 \leq 3.2^{-(m+3)} < 2^{-m}$.
\end{proof}

Now, we give a proof for Lemma \ref{lem:characteristicfunctionpspacerapid}.

\begin{proof}[Proof of Lemma \ref{lem:characteristicfunctionpspacerapid}]
	The major difficulty in directly approximating $\lVert
        A_n^{\chi_{\sigma}} - \mu(\sigma) \rVert_1$ is that for any
        $n,m \in \N$, $A_n^{\chi_{\sigma}}$ and $A_m^{\chi_{\sigma}}$
        may not be \textit{independent}. In order to overcome this, we
        use constructions similar to those used in proving
        Pillai's theorem (see \cite{Pillai40}, \cite{maxfield1952short} for normal numbers, \cite{nandakumar2019analogue} for continued fractions) in order to approximate each $A_n^{\chi_{\sigma}}$ with sums of \textit{disjoint} averages. These \textit{disjoint} averages turns out to be averages of independent random variables. Hence, elementary results from probability theory regarding independent random variables can be used to show that $A_n^{\chi_{\sigma}}$ converges to $\int f d\mu$ sufficiently fast.
	
	Observe that for any $x \in \Sigma^\infty$
	\begin{align*}
	A_n^{\chi_{\sigma}}(x)=\frac{|\{i \in [0,n-1] \mid T^{i}x \in [\sigma]\}|} {n}	
	\end{align*}
	Let $k=\lvert \sigma \rvert$. As in the proof of Theorem 3.1 from \cite{nandakumar2019analogue}, the following is a decomposition of the above term as \textit{disjoint} averages,
	\begin{align*}
	\frac{|\{i \in [0,n-1] \mid T^{i}x \in [\sigma]\}|} {n} =
		g_1(n) + g_2(n) +\dots+
		g_{(1+\left \lfloor\log_2  \frac{n}{k} \right \rfloor)}(n) + {\frac{(k-1). O(\log n) }{n}}	
	\end{align*}
	where,
	\begin{align*}
		g_{p}(n)= \begin{cases}
		n^{-1}	|\{i~ \mid~ T^{ki}x ~\in~ [\sigma]
				~,~ 0 ~\leq ~i ~\leq ~\lfloor n/k \rfloor \}| , ~\text{if} ~p=1\\
			n^{-1} \sum_{j=1}^{k-1} |\{i~ \mid~ T^{(2^{p-1})ki}x ~\in~ [S_j]
				~,~ 0 ~\leq ~i ~\leq \lfloor n/2^{p-1}k \rfloor \}| ,\text{ if } 1<p \leq
			(1+\left \lfloor\log_2(n/k)\right \rfloor)\\ 
			0, \text{ otherwise}
		\end{cases}
	\end{align*}
	where $S_j$ is the finite collection of $2^{(p-1)} k$ length blocks
	s.t $\sigma$ occurs in it at starting position $(2^{(p-2)}k - j +1)^{th}$
	position i.e $S_{j}$ is the set of strings of the form,
	$u\ a_1a_2 \dots a_k\ v$ where $u$ is some string of length $2^{p-2}k-j$, and $v$ is some string of length $2^{p-2}k-k+j$. 
	
	When $p=1$,
	\begin{align*}
	g_{1}(n)=\frac{\sum\limits_{i=1}^{\lfloor \frac{n}{k} \rfloor} X_i^{1,1}}{n}
	\end{align*}
	where,
	\begin{align*}
	X_{i}^{1,1}(x)=
	\begin{cases}
		1 \text{ if } x[ik+1,(i+1)k]=\sigma\\
		0 \text{ otherwise}
	\end{cases}
	\end{align*}
	When $1 < p \leq \lfloor \log_2(n/k) \rfloor$,
	\begin{align*}
	g_{p}(n)=\frac{\sum\limits_{i=1}^{\lfloor \frac{n}{2^{p-1}k} \rfloor}\sum\limits_{j=1}^{k-1} X_i^{p,j}}{n}
	\end{align*}
	where,
	\begin{align*}
	X_{i}^{p,j}(x)=
	\begin{cases}
		1 \text{ if } x[2^{p-2}k-j+1,2^{p-2}k-j+k] = \sigma\\
		0 \text{ otherwise}
	\end{cases}
	\end{align*}
	Hence, 
	\begin{align*}
		A_n^{\chi_{\sigma}}(x) = \frac{\sum\limits_{i=1}^{\lfloor \frac{n}{k}\rfloor} X^{1,1}_i(x)}{n}  + \sum\limits_{p=2}^{\lfloor \log_2(\frac{n}{k}) \rfloor} \sum\limits_{j=1}^{k-1} \frac{\sum\limits_{i=1}^{\lfloor \frac{n}{2^{p-1}k} \rfloor} X_i^{p,j}}{n}  + \frac{(k-1). O(\log n) }{n}
	\end{align*}
	An important observation that we use later in the proof is that for any fixed $p$ and $j$, $\{X_i^{p,j}\}_{i=1}^{\infty}$ is a collection of i.i.d Bernoulli random variables such that $\mu(\{x:X_i^{p,j}(x)=1\})=2^{-\lvert \sigma \rvert}$. We show that the conclusion of the lemma holds when $q_1(m)=2(m+6)$ and $q_2(m)=5(m+6)$. For any $m \in \N$,
	\begin{multline}
	\label{eq:pspacerapiderrorterm1}
	\left\lVert \sum\limits_{p=m+5+2}^{\infty} \sum\limits_{j=1}^{k-1} \frac{1}{n}\sum\limits_{i=1}^{\lfloor \frac{n}{2^{p-1}k} \rfloor} X_i^{p,j} \right\rVert_2 \leq \sum\limits_{p=m+5+2}^{\infty} \frac{1}{2^{p-1}} \leq  \frac{1}{2^{m+5}}\\
	\end{multline}
	And for $n \geq \lvert \sigma \rvert^3 2^{q_2(m)}> \lvert \sigma \rvert^2 2^{2(m+5)}$,
	\begin{multline}
	\label{eq:pspacerapiderrorterm2}
		\left\lVert \frac{(k-1)O(\log(n))}{n} \right\rVert_2 = \left\lVert \frac{(k-1)O(\log(n))}{\sqrt{n}\sqrt{n}} \right\rVert_2 \leq \left\lvert \frac{k-1}{\sqrt{n}} \right\rvert \leq \left\lvert \frac{k-1}{k 2^{m+5}}\right\rvert \leq \frac{1}{2^{m+5}} 
	\end{multline}
	Let,
	\begin{align*}
	D_{n,m}^\sigma (x) = \frac{\sum\limits_{i=1}^{\lfloor \frac{n}{k}\rfloor} X^{1,1}_i(x)}{n}  + \sum\limits_{p=2}^{m+5+2} \sum\limits_{j=1}^{k-1} \frac{\sum\limits_{i=1}^{\lfloor \frac{n}{2^{p-1}k} \rfloor} X_i^{p,j}}{n}	
	\end{align*}
	From (\ref{eq:pspacerapiderrorterm1}) and (\ref{eq:pspacerapiderrorterm2}), we get that
	\begin{align*}
	\lVert A_n^{\chi_\sigma}-D_{n,m}^\sigma \rVert_2 \leq \frac{2}{2^{m+5}}. 	
	\end{align*}
	Let,
	\begin{align*}
	E_{n,m}^\sigma (x) = \left( \frac{\sum\limits_{i=1}^{\lfloor \frac{n}{k}\rfloor} X^{1,1}_i(x)}{\lfloor \frac{n}{k} \rfloor} -\frac{1}{2^k} \right) \frac{\lfloor \frac{n}{k} \rfloor}{n} + \sum\limits_{p=2}^{m+5+2} \sum\limits_{j=1}^{k-1} \left( \frac{\sum\limits_{i=1}^{\lfloor \frac{n}{2^{p-1}k} \rfloor} X_i^{p,j}}{\lfloor \frac{n}{2^{p-1}k} \rfloor} -\frac{1}{2^k} \right) \frac{\lfloor \frac{n}{2^{p-1}k} \rfloor}{n}
	\end{align*}
	Now,
	\begin{align*}
	D_{n,m}^\sigma (x) - E_{n,m}^\sigma (x) = \frac{1}{2^k k} + 	\sum\limits_{p=2}^{m+5+2} \sum\limits_{j=1}^{k-1} \frac{1}{2^k}\frac{\lfloor \frac{n}{2^{p-1}k} \rfloor}{n}
	\end{align*}
	It follows that,
	\begin{align*}
	\lVert D_{n,m}^\sigma (x) - E_{n,m}^\sigma \rVert_2 &\leq \frac{1}{2^k} + 	\sum\limits_{p=2}^{m+5+2} \sum\limits_{j=1}^{k-1} \frac{1}{2^k 2^{p-1}k} \\
	&\leq \frac{1}{2^k} + \sum\limits_{p=2}^{m+5+2} \frac{1}{2^k 2^{p-1}} \\
	&\leq \frac{1}{2^k} + \sum\limits_{p=2}^{m+5+2} \frac{1}{2^k} \\
	&\leq \frac{m+5+2}{2^{k}}
	\end{align*}
	Hence, if $\lvert \sigma \rvert = k \geq q_1(m)=m+5+2+m+5$ then,
	\begin{align*}
	\lVert D_{n,m}^\sigma (x) - E_{n,m}^\sigma \rVert_2 \leq \frac{1}{2^{m+5}}
	\end{align*}
	and,
	\begin{align*}
	\lVert A_n^{\chi_\sigma}-\mu(\sigma) \rVert_2 &\leq 	\lVert A_n^{\chi_\sigma}-D_{n,m}^\sigma \rVert_2 + \lVert D_{n,m}^\sigma (x) - E_{n,m}^\sigma \rVert_2 + \lVert E_{n,m}^{\sigma} \rVert_2 + \frac{1}{2^k} \\
	&\leq \frac{3}{2^{m+5}}+\lVert E_{n,m}^{\sigma} \rVert_2+\frac{1}{2^{2m+12}} \\
	&\leq \frac{4}{2^{m+5}} + \lVert E_{n,m}^{\sigma} \rVert_2.
	\end{align*}
	Hence, in order to show that for all $n \geq \lvert \sigma\rvert^3 2^{q_2(m)}$, $\lVert A_n^{\chi_\sigma}-\mu(\sigma) \rVert_1 \leq \lVert A_n^{\chi_\sigma}-\mu(\sigma) \rVert_2 \leq 2^{-m}$, it is enough to show that for all $n \geq \lvert \sigma\rvert^3 2^{q_2(m)}$, $\lVert E_{n,m}^{\sigma} \rVert_2 \leq 2^{-(m+5)}$. Observe that,
	\begin{align*}
		\lVert E_{n,m}^{\sigma} \rVert_2 \leq \left\lVert \frac{1}{\lfloor \frac{n}{k} \rfloor} \sum\limits_{i=1}^{\lfloor \frac{n}{k}\rfloor} X^{1,1}_i(x) -\frac{1}{2^k} \right\rVert_2 + \sum\limits_{p=2}^{m+5+2} \sum\limits_{j=1}^{k-1} \left\lVert \frac{1}{\lfloor \frac{n}{2^{p-1}k} \rfloor}\sum\limits_{i=1}^{\lfloor \frac{n}{2^{p-1}k} \rfloor} X_i^{p,j} -\frac{1}{2^k} \right\rVert_2.
	\end{align*}
	Let $Y_1,Y_2,\dots Y_n$ be i.i.d Bernoulli random variables, 
	\begin{align*}
	\left\lVert \frac{1}{n} \sum\limits_{i=1}^{n} Y_i - \mathbf{E}(Y_1) \right\rVert_2	&= \sqrt{\mathbf{E}\left(\left( \frac{1}{n}\sum\limits_{i=1}^{n}Y_i - \mathbf{E}(Y_1)\right)^2\right)} \\
	&= \sqrt{\mathrm{Var}\left( \frac{1}{n}\sum\limits_{i=1}^{n}Y_i\right)}\\
	&= \sqrt{\frac{1}{n^2}n\mathrm{Var}(Y_1)}\\
	&\leq \frac{\sqrt{\mathrm{Var}(Y_1)}}{\sqrt{n}} \\
	&\leq \frac{1}{2\sqrt{n}}
	\end{align*}
	The last inequality follows from the fact that the variance of Bernoulli random variables are always bounded by $\frac{1}{4}$. Hence, if $n \geq \lvert \sigma\rvert^3 2^{q_2(m)} = \lvert \sigma\rvert^3 2^{5(m+6)}$ then,
	\begin{align*}
	\left\lfloor \frac{n}{k} \right\rfloor > k^2 2^{4(m+6)}	
	\end{align*}
	and
	\begin{align*}
	\left\lfloor \frac{n}{2^{p-1}k} \right\rfloor \geq \frac{k^3 2^{5(m+6)}}{2^{m+5+1}k} > k^2 2^{4(m+6)}.
	\end{align*}
	Hence for all $n \geq \lvert\sigma\rvert^3 2^{q_2(m)} = \lvert \sigma\rvert^3 2^{5(m+6)}$,
	\begin{align*}
	\lVert E_{n,m}^{\sigma} \rVert_2 &\leq \frac{1}{2k2^{2(m+6)}} + (m+6)k \frac{1}{2k 2^{2(m+6)}}\\
	&< \frac{1}{2^{m+6}} + \frac{1}{2^{m+6}} \\
	&\leq \frac{1}{2^{m+5}}.
	\end{align*}
	Hence we obtain the desired conclusion.
\end{proof}

We remark that since Lemma \ref{lem:characteristicfunctionpspacerapid} is true with the $L^1$-norm replaced by the $L^2$-norm, Theorem \ref{thm:l1functionpspacerapid} is also true in the $L^2$ setting. i.e, if a function $f$ is $\PSPACE$ $L^2$-computable (replacing $L^1$ norms with $L^2$ norms in definition \ref{def:pspacel1computablefunction}) then there exists a polynomial $q$ satisfying the following: given any $m\in \N$, for all $n \geq 2^{q(m)}$, $\lVert A_n^f -\int f d\mu \rVert_2 \leq 2^{-m}$. Hence, for $\PSPACE$ $L^2$-computable functions and the left shift transformation $T$, we get bounds on the convergence speed in the von-Neumann's ergodic theorem.
%
%

We now show that $\PSPACE$ ergodicity is a stronger version of $\ln^2$-ergodicity introduced in \cite{galatolo2009constructive}. 

\begin{lemma}
\label{lem:exprapidequivalence}
Let $T:\Sigma^\infty \to \Sigma^\infty$ be any measurable transformation. $T$ is $\PSPACE$ ergodic if and only if for any $f\in L^{\infty}(\Sigma^\infty,\mu)$, there exist $c>0$ and $k \in \N$ such that for all $n>0$,
	\begin{align*}
	\left\lvert \frac{1}{n}\sum\limits_{i=0}^{n-1}\int f\circ T^i. f - \int f d\mu \int f d\mu \right\rvert \leq \frac{c}{2^{(\ln n)^\frac{1}{k}}}.	
	\end{align*}
\end{lemma}
\begin{proof}
	We prove the forward implication first. The proof uses techniques from the proof of Theorem 4 in \cite{galatolo2009constructive}. From the hypothesis there exist $c>0$ and $k \in \N$ such that for all $n>0$,
	\begin{align*}
	\left|\frac{1}{n} \sum\limits_{i=0}^{n-1} f\circ T^i.f - \left(\int f d\mu\right)^2 \right| < \frac{c}{2^{(\ln n)^\frac{1}{k}}}
	\end{align*}
	By replacing $f$ with $f - \int f d\mu$, without loss of generality we assume that $\int f d\mu=0$. First we show that $A_n^f$ is $\PSPACE$-rapid almost everywhere convergent to $\int f d\mu$. Following the steps in proof of Lemma 6 in \cite{galatolo2009constructive} we get that,
	\begin{align*}
	\mu \left( \left\{x: \left| A_n^f-\int f d\mu \right| > \frac{1}{2^{m_1}} \right\} \right)	 \leq 2^{2m_1}\left(\frac{\Vert f\rVert_\infty^2}{n}+\frac{c}{2^{(\ln n)^\frac{1}{k}}}\right)
	\end{align*}
	Hence, there exists a polynomial $q$ such that,
	\begin{align*}
		\mu \left( \left\{x: \sup\limits_{n \geq 2^{q(m_1,m_2)}} \left| A_{\sqrt{n}}^f-\int f d\mu \right| > \frac{1}{2^{m_1}} \right\} \right)	 \leq \frac{1}{2^{m_2}} 
	\end{align*}
	From the proof of Lemma 7 in \cite{galatolo2009constructive}, we get that for $m$ such that $\sqrt{n} \leq m \leq \sqrt{n+1}$ and $\beta_n = \frac{\sqrt{n}}{\sqrt{n+1}}$,
	\begin{align*}
	\lVert A_{\sqrt{n}}^f-A_m^f\rVert_\infty \leq 2 (1-\beta_n)\lVert f\rVert_\infty	
	\end{align*}
	Let $l_1,l_2$ be any numbers such that $2^{l_1}\geq \lVert f \rVert_\infty$ and $2^{l_2} \geq 2k$. Let $p(n)=2n-2$. It is easy to see that for all $n \geq 2^{p(m +l_1)}$,
	\begin{align*}
		\frac{\sqrt{n}}{\sqrt{n+1}} &= \sqrt{1-\frac{1}{n+1}}
		\\&\geq 1-\frac{1}{2^{m_1+l_1+2}}\\
		&\geq 1-\frac{1}{2^{m_1+2}\lVert f\rVert_\infty}
	\end{align*}

	From the two previous inequalities we get that for $n\geq 2^{p(m +l_1)}$ and $m$ such that $\sqrt{n} \leq m \leq \sqrt{n+1}$,
	\begin{align*}
	\lVert A_{\sqrt{n}}^f-A_m^f\rVert_\infty \leq \frac{1}{2^{m_1 +1}}	
	\end{align*}
	Hence,
	\begin{align*}
	\bigcup\limits_{\sqrt{n}\leq m \leq \sqrt{n+1}} \left( \left\{x: \left| A_m^f-\int f d\mu \right| > \frac{1}{2^{m_1}} \right\} \right)	 \subseteq \left( \left\{x: \left| A_{\sqrt{n}}^f-\int f d\mu \right| > \frac{1}{2^{m_1+1}} \right\} \right)
	\end{align*}
	Let $r(m_1,m_2)=q(m_1+1,m_2)+p(m +l_1)$. Now,
	\begin{align*}
		\mu \left( \left\{x: \sup\limits_{n \geq 2^{r(m_1,m_2)}} \left| A_{n}^f-\int f d\mu \right| > \frac{1}{2^{m_1}} \right\} \right)	 &\leq \mu \left( \left\{x: \sup\limits_{n \geq 2^{q(m_1,m_2)}} \left| A_{\sqrt{n}}^f-\int f d\mu \right| > \frac{1}{2^{m_1}} \right\} \right)	 \\
		&\leq \frac{1}{2^{m_2}}
	\end{align*}
Hence $A_n^f$ is $\PSPACE$-rapid almost everywhere convergent to $\int f d\mu$. Now, given any $i$, for $n \geq 2^{r(i+1,i+l_1+1)}$,
\begin{align*}
\lVert A_n^f-\int f d\mu\rVert_1 \leq \frac{1}{2^{i+1}} + \frac{\lVert f\rVert_\infty}{2^{i+l_1+1}} \leq \frac{1}{2^i}	
\end{align*}
Hence, $\int f d\mu$ is a $\PSPACE$-rapid $L^1$-limit point of $A_n^f$.

Now, we prove the backward direction. Observe that,
\begin{align*}
\left\lvert \frac{1}{n}\sum\limits_{i=0}^{n-1}\int f\circ T^i. f - \int f d\mu \int f d\mu \right\rvert &= 	\left\lvert \frac{1}{n}\sum\limits_{i=0}^{n-1}\int f\circ T^i. f -  \int \left(\int f d\mu \right) f d\mu \right\rvert \\
&= \left\lvert \frac{1}{n}\sum\limits_{i=0}^{n-1}\int \left(f\circ T^i -   \int f d\mu \right) f d\mu \right\rvert \\
&= \left\lvert \int \left(\frac{1}{n}\sum\limits_{i=0}^{n-1}f\circ T^i -   \int f d\mu \right) f d\mu \right\rvert \\
&\leq \lVert A_n^f - \int f d\mu\rVert_2 \lVert f\rVert_2
\end{align*}

The last inequality follows from the Cauchy-Schwarz inequality. From Theorem \ref{thm:aeconvergence} we get that $A_n^f$ is $\PSPACE$-rapid almost everywhere convergent to $\int f d\mu$. Hence, there exists a polynomial $p$ such that for any $m_1$ and
$m_2$,
\begin{align*}
  \mu\left(\left\{ x:\sup\limits_{n \geq 
  2^{p(m_1+m_2)}}
  \left|A^f_n(x)-\int f d\mu \right| > \frac{1}{2^{m_1}}\right\} \right)	\leq
  \frac{1}{2^{m_2}}. 
 \end{align*}
  Let $l_1$ be any number such that $2^{l_1} \geq \lVert f \rVert_\infty \geq \lVert f \rVert_2$. Now, for any given $m>0$, for all $n \geq 2^{p(2m+2l_1+1,2m+5l_1+1)}$,
  \begin{align*}
   \lVert A_n^f -\int f d\mu \rVert_2^2 \leq \frac{1}{2^{2m+2l_1+1}}+\frac{\lVert f\rVert^2_\infty}{2^{2m+5l_1+1}} \leq \frac{1}{2^{2m+2l_1}}.
  \end{align*} 
  Hence, there exists $j>0$ such that for any given $m>0$, for all $n \geq 2^{(m+l_1)^j}$
  \begin{align*}
   \lVert A_n^f -\int f d\mu \rVert_2 \leq \frac{1}{2^{m+l_1}}.
  \end{align*} 
  Now given $n>0$, let $m$ be any number such that $2^{(m+l_1)^j} \leq n \leq 2^{(m+l_1+1)^j}$,
\begin{align*}
\left\lvert \frac{1}{n}\sum\limits_{i=0}^{n-1}\int f\circ T^i. f - \int f d\mu \int f d\mu \right\rvert &\leq \frac{\lVert f\rVert_2}{2^{m+l_1}} \leq \frac{1}{2^m}
\end{align*}
Since $n \leq 2^{(m+l_1+1)^j}$,
\begin{align*}
\left\lvert \frac{1}{n}\sum\limits_{i=0}^{n-1}\int f\circ T^i. f - \int f d\mu \int f d\mu \right\rvert &\leq \frac{2^{l_1+1}}{2^{(\ln n)^\frac{1}{j}}}
\end{align*}
The result follows with $c=2^{l_1+1}$.
\end{proof}

\section{$\PSPACE$-rapid almost everywhere convergence of ergodic
  averages} In the earlier section, we related the rapidity of $L^1$
convergence of $f_n$ to $f$, to the $L^1$ convergence speed of $A^f_n$
to $\int f$.  Now we present $\PSPACE$ versions of Theorem 2 and
Proposition 5 from \cite{hoyruprojasgalatolo}, relating the $L^1$
convergence of $A^f_n$ to $\int f$ to its almost everywhere
convergence.  The main estimate which we require in this section is
the maximal ergodic inequality, which we now recall.

\begin{lemma}[Maximal ergodic inequality \cite{Bill86}]
If $f \in L^{1}(\Sigma^\infty,\mu)$ and $\delta>0$ then,
\begin{align*}
\mu\left(\left\{x:\sup\limits_{n \geq 1} |A^f_n(x)|>
\delta\right\}\right) \leq \frac{\lVert f\rVert_1}{\delta}.	
\end{align*}	
\end{lemma}

Using this lemma, we now prove the almost everywhere convergence of
ergodic averages. In contrast to \cite{hoyruprojasgalatolo}, we give a
direct proof of the theorem for $L^1$ functions with possibly infinite
essential supremum using Markov's inequality.

\begin{theorem}
\label{thm:aeconvergence}
Let $f$ be any function in $ L^{1}(\Sigma^\infty,\mu)$ and let $T$ be a measure preserving transformation. If $\int f d\mu$ is a
$\PSPACE$-rapid $L^1$-limit point of $A^f_n$ then $A_{n}^{f}$ is
$\PSPACE$-rapid almost everywhere convergent to $\int f d\mu$.
\end{theorem}
\begin{proof}
By replacing $f$ with $f-\int f d\mu$ we can assume without loss
of generality that $\int f d\mu =0$.
	
We construct a polynomial $q$ such that for any $m_1$ and
$m_2$,
\begin{align*}
  \mu\left(\left\{ x:\sup\limits_{n \geq 2^{
  q(m_1+m_2)}}
  |A^f_n(x)| > 2^{-m_1}\right\} \right)	\leq
  2^{-m_2}. 
\end{align*}
Since, $\int f d\mu = 0$ is a $\PSPACE$-rapid $L^1$-limit point of $A^f_n$ there is a polynomial $p$ such that $\exists k \leq 2^{p(m_1 +m_2 +2)}$ with
$\lVert A^f_k\rVert_1 \leq \frac{1}{2^{m_1 + m_2 + 2}}$. 
	
Applying the maximal ergodic inequality to $g=A^f_k$, we get
\begin{align}
  \label{eq:maximalergodicapplication}
  \mu\left(\left\{ x:\sup\limits_{n \geq 1} |A^g_n(x)| >
  \frac{1}{2^{m_1 +1}}\right\} \right)	\leq \frac{1}{2^{m_2 +1}}. 
\end{align}
Expanding $A^g_n$,
\begin{align*}
  A^g_n = A^f_n + \frac{u\circ T^n -u}{nk},
\end{align*}
where $u=(k-1)f+(k-2)f\circ T+\dots + f\circ T^{k-2}$. Note that 
\begin{align*}
  \lVert u\rVert_1 \leq \frac{k(k-1)}{2} \lVert f\rVert_1.	
\end{align*}
Let $M$ be any upper bound for $\lVert f\rVert_1$. And let $n_0 = (2^{p(m_1
+m_2 +2)}-1)M 2^{m_1+m_2+2}$. From the above, we get $\lVert
A^g_n-A^f_n\rVert_1 \leq \frac{1}{2^{m_1 +m_2 +2}}$ for any $n \geq n_0$. Now from Markov's inequality it follows that,
\begin{align}
\label{eq:markovinequalityapplication}
	\mu\left(\left\{ x:\sup\limits_{n \geq n_0} |A^f_n(x)-A^g_n(x)| > \frac{1}{2^{m_1+1}}\right\} \right)	\leq \frac{\lVert A^g_n-A^f_n \rVert_1}{2^{m_1+1}} \leq \frac{1}{2^{m_2+1}}
\end{align}

Hence, from \ref{eq:maximalergodicapplication} and \ref{eq:markovinequalityapplication}, we get
\begin{align*}
  \mu\left(\left\{ x:\sup\limits_{n \geq n_0} |A^f_n(x)| > \frac{1}{2^{m_1}}\right\} \right)	\leq \frac{1}{2^{m_2}}.
\end{align*}
Since $n_0$ is upper bounded by a term of the form $2^{
  q(m_1+m_2)}$ for a polynomial $q$, the claim follows.
\end{proof}

If $f \in L^\infty$, the converse of Theorem \ref{thm:aeconvergence} can be easily obtained by expanding $\lVert A_n^f - \int f d\mu \rVert_1$.

\section{An ergodic theorem for $\PSPACE$ $L^1$ functions}
\label{sec:pspaceergodictheorem}
We now establish the main theorem in our work, namely, that for
$\PSPACE$ $L^1$ computable functions, the ergodic average exists, and
is equal to the space average, on every $\EXP$ random. We utilize
the almost everywhere convergence results proved in the previous
section, to prove the convergence on every $\PSPACE$/$\EXP$ random. The convergence notions involved in proving the $\PSPACE$/$\SUBEXP$-space ergodic theorems and their interrelationships are summarized in Figure \ref{fig:convergencenotions}.
\begin{figure}[h!]
\begin{center}
\[
\begin{tikzcd}
& & \boxed{\PSPACE \text{ ergodic theorem}(\ref{thm:pspaceergodictheorem})}\\
\boxed{A_n^f \dashrightarrow \int f d\mu}  \arrow [bend left=60]{r}{\text{Theorem } \ref{thm:aeconvergence}}& \boxed{A_n^f \xrightarrow[\text{ a.e}]{\PSPACE} \int fd\mu} \arrow[bend left=60]{l}{f \in L^\infty}\arrow[sloped,auto,swap]{ur}{f \in \PSPACE \text{ }L^1}\arrow[sloped,above]{dr}{f \in \SUBEXP \text{ }L^1}\\
& & \boxed{\SUBEXP \text{ ergodic theorem}(\ref{thm:subexpergodictheorem})}\\ 
\end{tikzcd}
\]
\end{center}
\caption{Relationships between the major convergence notions involving $\PSPACE$ simple measure preserving transformations. $A_n^f \dashrightarrow \int f d\mu$ denotes that $\int f d\mu$ is a $\PSPACE$-rapid $L^1$-limit point of $A_n^f$. $\PSPACE$/$\SUBEXP$-space ergodicity is required only for obtaining the ergodic theorems from $\PSPACE$ a.e convergence.}
\label{fig:convergencenotions}
 \end{figure}
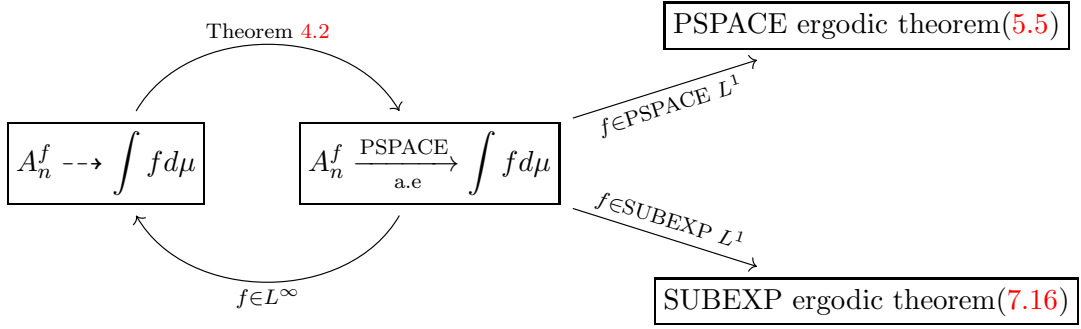

The following fact was shown in \cite{HuaStu16}. However, for our ergodic theorem we require an alternate proof of this fact using techniques from \cite{RuteThesis}.

\begin{lemma}
\label{lem:convergencelemma}
Let $\<f_n\>_{n=1}^{\infty}$ be a $\PSPACE$ sequence of simple
functions which converges $\PSPACE$-rapid almost everywhere to $f \in L^1
(\Sigma^\infty,\mu)$. Then,
\begin{enumerate}
	\item $\lim\limits_{n \to \infty}f_n(x)$ exists for all $\EXP$ random $x$.
	\item Given a $\PSPACE$ sequence of simple functions $\<g_n\>_{n=1}^{\infty}$ which is $\PSPACE$-rapid almost everywhere convergent to $f$, $\lim\limits_{n \to \infty}g_n(x)=\lim\limits_{n \to \infty}f_n(x)$ for all $\EXP$ random $x$. 
\end{enumerate}	
\end{lemma}
\begin{proof}
We initially show 1. For each $k \geq 0$, since $\<f_n\>_{n=1}^{\infty}$ is $\PSPACE$-rapid almost everywhere convergent to $f$, we have a polynomial $q$ such that
\begin{align*}
\mu \left( \left\{x : \sup\limits_{n \geq 2^{q(k)}} |f_n(x)-f(x)| \geq \frac{1}{2^{k+2}} \right\} \right) \leq \frac{1}{2^{k+2}}.
\end{align*}
	It is easy to verify that
\begin{align*}
\mu \left( \left\{x : \sup\limits_{n \geq 2^{q(k)}} |f_n(x)-f_{2^{q(k)}}(x)| \geq \frac{1}{2^{k+1}} \right\} \right) \leq \frac{1}{2^{k+1}}.
\end{align*}	
Define
\begin{align*}
U_k = \left\{x:\max\limits_{2^{q(k)} \leq n \leq 2^{q(k+1)}} |f_n(x)-f_{2^{q(k)}}(x)| \geq \frac{1}{2^{k+1}}\right\}.
\end{align*}
Observe that
\begin{align*}
\mu(U_k) \leq \mu \left( \left\{x : \sup\limits_{n \geq 2^{q(k)}}
|f_n(x)-f_{2^{q(k)}}(x)| \geq \frac{1}{2^{k+1}} \right\} \right) \leq
\frac{1}{2^{k+1}}.	 
\end{align*}
Let $r$ be the controlling polynomial and let $M$ be the $\PSPACE$
machine witnessing the fact that $\<f_n\>_{n=1}^{\infty}$ is a
$\PSPACE$ sequence of simple functions. $U_k$ is hence a union of
cylinders of length at most $r(2^{q(k+1)})$. Let the  machine $N$ on input $(\sigma,1^k)$ do the following:

\begin{enumerate}
	\item If $\lvert \sigma \rvert < r(2^{q(k+1)})$ then, output $0$.
	\item Compute $f_{2^{q(k)}}(\sigma0^\infty)$ by running $M(1^{2^{q(k)}}, \sigma) $ and store the result. 
	\item For each $n \in [2^{q(k)},2^{q(k+1)}]$ do the following:
	\begin{enumerate}
	\item Compute $f_{n}(\sigma0^\infty)$ by running $M(1^{n}, \sigma) $ and store the result.
	\item Check if $|f_n(\sigma0^\infty)-f_{2^{q(k)}}(\sigma0^\infty)| \geq \frac{1}{2^{k+1}}$. If so, output $1$.
	\end{enumerate}
	\item Output $0$.
\end{enumerate}
Since $N$ rejects any $\sigma$ with length less than $r(2^{q(k+1)})$, the simulation of $M(1^{2^{q(k)}}, \sigma)$ is always a polynomial space operation. Hence, $N$ is an $\EXP$-time machine witnessing the fact that $\<U_k\>_{k=1}^{\infty}$ is an $\EXP$ sequence of open sets.

 Define,
\begin{align*}
V_k = \bigcup\limits_{i=k}^{\infty} U_i.	
\end{align*}
Since $\mu(U_k)\leq 2^{-(k+1)}$, machine $N$ above can be easily modified to show that $\<V_k\>_{k=1}^{\infty}$ is an $\EXP$ test.

If $x \in \Sigma^\infty$ is an $\EXP$ random then $x$ is in at most
finitely many $V_k$ and hence in only finitely many $U_k$. Hence, for large enough $k$ and for all $n \geq 2^{q(k)}$ we have 
\begin{align*}
\left|f_n(x)-f_{2^{q(k)}}(x)\right|	 \leq \sum\limits_{j=k}^{\infty}
\frac{1}{2^{j+1}} \leq \frac{1}{2^k}.
\end{align*}
This shows that $f_n(x)$ is a Cauchy sequence. This completes the
proof of 1. 

Given $\<g_n\>_{n=1}^{\infty}$ which is $\PSPACE$-rapid almost everywhere
convergent to $f$, the interleaved sequence $f_1,g_1,f_2,g_2,f_3,g_3
\dots$ can be easily verified to be $\PSPACE$-rapid almost everywhere
convergent to $f$. 2 now follows directly from 1.  
\end{proof}

The following immediately follows from the above lemma.
\begin{corollary}
\label{cor:convergencelemmacorollary1}
Let $f \in L^1 (\Sigma^\infty,\mu)$ be a $\PSPACE$ $L^1$-computable
function with an $L^1$ approximating $\PSPACE$ sequence of simple functions
$\<f_n\>_{n=1}^{\infty}$. Then,
\begin{enumerate}
\item $\lim\limits_{n \to \infty}f_n(x)$ exists for all $\EXP$
  random $x$. 
\item Given a $\PSPACE$ sequence of simple functions
  $\<g_n\>_{n=1}^{\infty}$ $L^1$ approximating $f$, $\lim\limits_{n \to
  \infty}g_n(x)=\lim\limits_{n \to \infty}f_n(x)$ for all $\EXP$
  random $x$. 
\end{enumerate}	
\end{corollary}
\begin{proof}
For
any $m_1,m_2 \geq 0$,
	\begin{multline*}
\mu \left( \left\{x : \sup\limits_{n \geq m_1+m_2+1} |f_n(x)-f(x)|
\geq \frac{1}{2^{m_1}} \right\} \right) \leq \sum\limits_{n=m_1 +m_2
  +1}^{\infty} \mu \left( \left\{x : |f_n(x)-f(x)| \geq
\frac{1}{2^{m_1}} \right\} \right) \\ 
\leq  \sum\limits_{n=m_1 +m_2 +1}^{\infty} \lVert f_n-f\rVert_1 2^{m_1}
\ \leq\  \sum\limits_{i=1}^{\infty} \frac{1}{2^{m_2+i}}
\ =\  2^{-m_2}.
\end{multline*}
Hence, $\<f_n\>_{n=1}^{\infty}$ is $\PSPACE$-rapid almost everywhere
convergent to $f$. The claim follows due to Lemma \ref{lem:convergencelemma}.
\end{proof}

The following properties satisfied by $\PSPACE$ simple transformations and $\PSPACE$ $L^1$-computable functions are useful in our
proof of the $\PSPACE$ ergodic theorem.
\begin{lemma}
\label{lem:pspacefunctionintegral}
Let $f$ be a $\PSPACE$ $L^1$-computable function. Let $I_f:\Sigma^\infty \to \Sigma^\infty$ be the constant function taking the value $\int f d\mu$ over all $x \in \Sigma^\infty$. Then, $I_f$ is $\PSPACE$ $L^1$-computable and $\widetilde{I_f}(x)=\int f d\mu$ for all $\EXP$ random $x$.
\end{lemma}
\begin{proof}
Let $\<f_n\>_{n=1}^{\infty}$ be a $\PSPACE$ sequence of simple functions, $M$ be a $\PSPACE$ machine and $p$ be a controlling polynomial witnessing the fact that $f$ is $\PSPACE$ $L^1$-computable. We construct a $\PSPACE$ sequence of simple functions $\<f'_n\>_{n=1}^{\infty}$ where each $f'_n$ is the constant function taking the value $\int f_n d\mu$. Since $ \lVert f'_n - I_f \rVert_1 =\lvert \int f_n d\mu - \int f d\mu \rvert \leq \lVert f_n -f \rVert_1 \leq  2^{-n}$, it follows that $I_f$ is $\PSPACE$ $L^1$-computable. Now, from part 2 of Lemma \ref{lem:convergencelemma}, we get that $\widetilde{I_f}(x)=\lim\limits_{n \to \infty}f'_n(x)=\lim\limits_{n \to \infty}\int f_n d\mu=\int f d\mu$ for all $\EXP$ random $x$.

On input $(1^n,\sigma)$, let machine $N$ do the following:
\begin{enumerate}
	\item Let $\mathrm{Sum}=0$.
	\item For each $\alpha \in \Sigma^{p(n)}$ do the following:
	\begin{enumerate}
		\item Run $M(1^n,\alpha)$ and add the result to $\mathrm{Sum}$.
	\end{enumerate}
	\item Output $\mathrm{Sum}/2^{p(n)}$.
\end{enumerate}
Let $t$ be a polynomial upper bound for the space complexity of $M$. Then, the result of $M(1^n,\alpha)$ is always upper bounded by $2^{t(n+p(n))}$ and representable in $t(n+p(n))$ space. The sum of at most $2^{p(n)}$ many such numbers is upper bounded by $2^{t(n+p(n))+p(n)}$ and representable in $t(n+p(n))+p(n)$ space. Dividing the running sum by $2^{p(n)}$ can also be done in polynomial space. Hence, $N$ is a $\PSPACE$ machine computing the sequence $\<f'_n\>_{n=1}^{\infty}$ where each $f'_n=\int f_n d\mu$.
\end{proof}
\begin{lemma}
\label{lem:pspacesequencelemma}
Let $f$ be a $\PSPACE$ $L^1$-computable function with an $L^1$ approximating
$\PSPACE$ sequence of simple functions $\<f_n\>_{n=1}^{\infty}$. Let
$T$ be a $\PSPACE$ simple transformation and $p$ be a
polynomial. Then, $\<A_n^{f_{p(n)}}\>_{n=1}^{\infty}$ is a $\PSPACE$
sequence of simple functions.
\end{lemma}
\begin{proof}
	Let $q$ be a controlling polynomial and $M_f$ be a machine
        witnessing the fact that $\<f_n\>_{n=1}^{\infty}$ is a
        $\PSPACE$ sequence of simple functions. Let $c_T$ be a controlling
        constant witnessing the fact that $T$
        is a $\PSPACE$ simple transformation. For any $n \geq 1$, we
        have 
	\begin{align*}
	A^{f_{p(n)}}_n = \frac{f_{p(n)}+f_{p(n)}\circ T+f_{p(n)}\circ
          T^2+\dots f_{p(n)}\circ T^n}{n}.
	\end{align*}
	The functions $\<f_{p(n)}\circ T^i\>_{i=1}^{n}$ are simple
        functions defined on cylinders of length at most $q(p(n))+c_T
        n$. Hence, the polynomial $r(n)=q(p(n))+c_T n$ is a
        controlling polynomial for the sequence of functions
        $\<A^{f_{p(n)}}_n\>_{n=1}^{\infty}$ as in condition 1 of
        Definition \ref{def:pspacesequenceofsimplefunctions}.  Now, let
        us verify condition 2 of Definition \ref{def:pspacesequenceofsimplefunctions}. Let $M$ be the machine from Lemma \ref{lem:pspacetransformationcomputation}.  We construct a
        machine $N$ such that for each $n \in \N$ and $\alpha \in
        \Sigma^*$, 
	\begin{align*}
		N(1^n,\alpha)=\begin{cases}
                A^{f_{p(n)}}_n(\alpha0^\infty) & \text{if } |\alpha| \geq r(n)\\
                ? & \text{otherwise.}
                \end{cases}
	\end{align*}
		On input $(1^n,\alpha)$ if $|\alpha|<r(n)$ then $N$ outputs $?$ else it operates as follows:
		\begin{enumerate}
			\item Let $\mathrm{Sum}=0$
			\item For each $i \in [1,n]$, do the following:
			\begin{enumerate}
				\item For each string $\sigma $ of length $q(p(n))$, do the following:
				\begin{enumerate}
				\item If $M(1^i,\sigma,\alpha)=1$, then let $\mathrm{Sum}=\mathrm{Sum}+M_f(1^{p(n)},\sigma)$.
				\end{enumerate}
			 \end{enumerate}
			 \item Output $\mathrm{Sum}/n$.
		\end{enumerate}
		
		 If $t_1$ is a polynomial upper bound for the space complexity of $M$ and $t_2$ is a polynomial upper bound for the space complexity of $M_f$ then, each $f_{p(n)}\circ T^i$ can be computed in $O(t_1(2q(p(n))+c_T n+n)+t_2(q(p(n))+p(n)))$ space for any $i \leq n$. The results of computations of $f_{p(n)}\circ T^i$ for $i\in [1,n]$ can be added up and divided by $n$ in polynomial space. $N$ outputs the result of this computation. Since $N$ is a $\PSPACE$ machine, the proof is complete.
\end{proof}

Now, we prove the ergodic theorem for $\PSPACE$ $L^1$ functions, which
is our main result. The proof involves adaptations of techniques from
Rute \cite{RuteThesis}, together with new quantitative bounds which
yield the result within prescribed resource bounds. 
\begin{theorem}
\label{thm:pspaceergodictheorem}
Let $T:(\Sigma^\infty,\mathcal{B}(\Sigma^\infty),\mu) \to (\Sigma^\infty,\mathcal{B}(\Sigma^\infty),\mu)$ be a $\PSPACE$ ergodic measure preserving
transformation. Then, for any $\PSPACE$ $L^1$-computable $f$, $\lim\limits_{n \to \infty}\widetilde{A^f_n}=\int f
d\mu$ on $\EXP$ randoms.
\end{theorem}
\begin{proof}
Let $\<f_m\>_{m=1}^{\infty}$ be any $\PSPACE$ sequence of simple
functions $L^1$ approximating $f$. We initially approximate $A_n^f$ with a
$\PSPACE$ sequence of simple functions $\<g_n\>_{n=1}^{\infty}$ which
converges to $\int f d\mu$ on $\EXP$ randoms. Then we show that
$\widetilde{A}^f_n$ has the same limit as $g_n$ on $\PSPACE$ randoms and hence on $\EXP$ randoms.  
	
For each $n$, it is easy to verify that $\<A^{f_m}_n\>_{m=1}^{\infty}$ is a $\PSPACE$
sequence of simple functions $L^1$ approximating $A^f_n$ with the same rate
of convergence. Using techniques similar to those in Lemma
\ref{lem:convergencelemma} and Corollary
\ref{cor:convergencelemmacorollary1}, we can obtain a polynomial $p$
such that
\begin{align*}
  \mu \left( \left\{x : \sup\limits_{m \geq p(n+i)}
  |A^{f_m}_n(x)-A^{f_{p(n+i)}}_n(x)| \geq \frac{1}{2^{n+i+1}} \right\}
  \right) \leq \frac{1}{2^{n+i+1}}. 
\end{align*}
For every $n>0$, let $g_n=A^{f_{p(n)}}_n$. We initially show that $\<g_n\>_{n=1}^{\infty}$ converges to
$\int f d\mu$ on $\EXP$ randoms. Let $m_1,m_2 \geq 0$. From
Theorem \ref{thm:aeconvergence}, $A^f_n$ is $\PSPACE$-rapid almost
everywhere convergent to $\int f d\mu$. Hence there is a polynomial
$q$ such that 
\begin{align*}
  \mu \left( \left\{x : \sup\limits_{n \geq 2^{q(m_1+m_2)}}
  |A^f_n(x)-\int f d\mu| \geq \frac{1}{2^{m_1+1}} \right\} \right)
  \leq \frac{1}{2^{m_2+1}}. 
\end{align*}
Let $N(m_1,m_2)=\max\{2m_1,2m_2,2^{q(m_1 +m_2)}\}$. Then,
\begin{align*}
  \sum\limits_{n \geq N(m_1,m_2)} \frac{1}{2^{k+1}} = \frac{1}{2^{N(m_1,m_2)}} \leq \min \left\{ \frac{1}{2^{m_1 +1}},\frac{1}{2^{m_2 +1}} \right\}.
\end{align*}
Now, we have  
\begin{align*}
  \mu \left( \left\{x : \sup\limits_{n \geq N(m_1,m_2)} |g_n-\int f
  d\mu| > \frac{1}{2^{m_1}} \right\} \right) &\leq \sum\limits_{n \geq
    N(m_1,m_2)} \mu \left( \left\{x :  |g_n-A^f_n(x)| >
  \frac{1}{2^{m_1+1}} \right\} \right)\\ 
  &+ \mu \left( \left\{x : \sup\limits_{n \geq 2^{q(m_1+m_2)}}
  |A^f_n(x)-\int f d\mu| \geq \frac{1}{2^{m_1+1}} \right\} \right) \\ 
  &\leq \sum\limits_{n \geq N(m_1,m_2)} \frac{1}{2^{n+1}} +
  \frac{1}{2^{m_2 +1}} \\ 
  &\leq \frac{1}{2^{m_2}}.
\end{align*}
Note that $N(m_1,m_2)$ is bounded by $2^{(m_1+m_2)^c}$ for some $c \in \N$. Hence,
$g_n$ is $\PSPACE$-rapid almost everywhere convergent to $\int f
d\mu$. From Lemma \ref{lem:pspacesequencelemma} it follows that
$\<g_n\>_{n=1}^{\infty}=\<A^{f_{p(n)}}_n\>_{n=1}^{\infty}$ is a
$\PSPACE$ sequence of simple functions (in parameter $n$). Let $I_f:\Sigma^\infty \to \Sigma^\infty$ be the constant function taking the value $\int f d\mu$ over all $x \in \Sigma^\infty$. From the above
observations and Lemma \ref{lem:convergencelemma} we get
that $\lim\limits_{n \to \infty}g_n (x)=\widetilde{I_f}(x)$ for any $x$ which
is $\EXP$ random. Now, from Lemma \ref{lem:pspacefunctionintegral}, we get that $\lim\limits_{n \to \infty}g_n (x)=\int f d\mu$ for any $x$ which
is $\EXP$ random.

We now show that $\lim\limits_{n \to \infty} \widetilde{A}^f_n =
\lim\limits_{n \to \infty} g_n$ on $\PSPACE$ randoms. Define 
\begin{align*}
U_{n,i} = \left\{x:\max\limits_{p(n+i) \leq m \leq p(n+i+1)}
|A^{f_m}_n(x)-A^{f_{p(n+i)}}_n(x)| \geq \frac{1}{2^{n+i+1}}\right\}. 	
\end{align*}
We already know $\mu(U_{n,i})\leq \frac{1}{2^{n+i+1}}$. $U_{n,i}$ can be shown to be polynomial space
approximable in parameters $n$ and $i$ in the following sense. There exists a sequence of sets of strings $\<S_{n,i}\>_{i,n \in \N}$ and polynomial $p$ satisfying the following conditions:
\begin{enumerate}
  \item $U_{n,i} = [S_{n,i}]$.
   \item There exists a \emph{controlling polynomial} $r$ such
    that $max\{|\sigma|:\sigma \in  S_{n,i})\}\leq r(n+i)$.
  \item The function $g:\Sigma^* \times 1^*\times 1^* \to \{0,1\}$
    such that
    \begin{align*}
      g(\sigma,1^n,1^i)=
      \begin{cases}
	1 &\text{if }\sigma \in S_{n,i}\\ 0 &\text{otherwise,}
      \end{cases}	
    \end{align*}
    is decidable by a $\PSPACE$ machine.
\end{enumerate}
The above claims can be established by using techniques similar to those in Lemma \ref{lem:pspacesequencelemma} and Lemma \ref{lem:convergencelemma}. We now show the construction of a machine $N$ computing the function $g$ above. Let $M_f$ be a computing machine and let $q$ be a controlling polynomial for $\<f_n\>_{n=1}^{\infty}$. Let $c$ be a controlling constant for $T$. Let $M'$ be the machine from Lemma \ref{lem:pspacetransformationcomputation}. Machine $N$ on input $(\sigma,1^n,1^i)$ does the following:
\begin{enumerate}
	\item If $|\sigma| > q(p(n+i+1))+cn$, then output $0$.
	\item Compute $A_n^{f_{p(n+i)}}(\sigma0^\infty)$ as in Lemma \ref{lem:pspacesequencelemma} by using $M_f$ and $M'$ and store the result.
	\item For each $m \in [p(n+i),p(n+i+1)]$ do the following:
	\begin{enumerate}
	\item Compute $A_n^{f_{m}}(\sigma0^\infty)$ as in Lemma \ref{lem:pspacesequencelemma} by using $M_f$ and $M'$ and store the result.
	\item Check if $|A_n^{f_{m}}(\sigma0^\infty)-A_n^{f_{p(n+i)}}(\sigma0^\infty)| \geq \frac{1}{2^{n+i+1}}$. If so, output $1$.
	\end{enumerate}
	\item Output $0$.
\end{enumerate}
It can be easily verified that $N$ is a $\PSPACE$ machine. $r(n+i)=q(p(n+i+1))+cn$ is the controlling polynomial for $\<U_{n,i}\>_{n,i \in \N}$. Now, define
\begin{align*}
V_m = \bigcup\limits_{\substack{n,i \geq 0 \\ n+i=m}}	 U_{n,i}.
\end{align*}
Note that,
\begin{align*}
\mu(V_m) \leq \frac{m}{2^m}.	
\end{align*}
It can be shown that for any $j$,
\begin{align*}
\sum\limits_{n > j} \frac{m}{2^m} = \frac{1}{2^{j-1}} + \frac{j}{2^j}.	
\end{align*}
Given any $k \geq 0$, let $p(k)=3(k+1)$. Hence, we have
\begin{align*}
\sum\limits_{n = p(k)+1}^{\infty} \frac{m}{2^m} &=
\frac{1}{2^{3(k+1)}} + \frac{3(k+1)}{2^{3(k+1)}}
\ <\  \frac{1}{2^{k+1}} +\frac{1}{2^{k+1}}\frac{3(k+1)}{2^{2(k+1)}} 
\ <\ \frac{2}{2^{k+1}} = \frac{1}{2^k}.
\end{align*}
The last inequality holds since $3(k+1)<2^{2(k+1)}$ for all $k \geq
0$. Since each $V_m$ is a finite union of sets from $\<U_{n,i}\>_{n,i \in \N}$, the machine computing $\<U_{n,i}\>_{n,i \in \N}$ can be easily modified to construct a machine witnessing
that $\<V_m\>_{m=1}^{\infty}$ is a $\PSPACE$ approximable sequence of
sets. From these observations, it follows that $\<V_m\>_{m=1}^{\infty}$ is a $\PSPACE$ Solovay test. Now,
let $x$ be a $\PSPACE$ random. $x$ is in at most finitely many $V_m$
and hence in at most finitely many $U_{n,i}$. Hence for some large
enough $N$ for all $n \geq N$, $i \geq 0$ and for all $m$ such that
$p(n+i) \leq m \leq p(n+i+1)$, we have
$|A^{f_m}_n(x)-A^{f_{p(n+i)}}_n(x)| <\frac{1}{2^{n+i+1}}$. It follows
that for all $n \geq N$ and for all $m \geq p(n)$ that,
\begin{align*}
  |A^{f_m}_n(x)-g_n(x)| &= |A^{f_m}_n(x)-A^{f_{p(n)}}_n(x)| 
  \ \leq\  \sum\limits_{i=0}^{\infty} \frac{1}{2^{n+i+1}} 
  \ \leq\  2^{-n}.
\end{align*}
Therefore, $\lim\limits_{n \to \infty} \widetilde{A}^f_n (x) =
\lim\limits_{n \to \infty} g_n(x)$ on all $\PSPACE$ random $x$ and
hence on all $x$ which is $\EXP$ random. 

Hence, we have shown that $\lim\limits_{n \to \infty} \widetilde{A}^f_n = \int f d\mu$ on $\EXP$ randoms which completes the proof of the theorem.
\end{proof}

\section{A partial converse to the $\PSPACE$ Ergodic Theorem}
In this section we give a partial converse to the $\PSPACE$ ergodic theorem (Theorem \ref{thm:pspaceergodictheorem}). We show that for any $\PSPACE$ null $x$, there exists a function $f$ and transformation $T$ satisfying all the conditions in Theorem \ref{thm:pspaceergodictheorem} such that $\widetilde{A^f_n}(x)$ does not converge to $\int f d\mu$.

Let us first observe that due to Corollary \ref{cor:convergencelemmacorollary1}, Theorem \ref{thm:pspaceergodictheorem} is equivalent to the following:
\begin{theorem*}
Let $T$ be a $\PSPACE$ ergodic measure preserving
transformation such that for any $\PSPACE$ $L^1$-computable $f$, $\int f d\mu$ is an $\PSPACE$-rapid $L^1$-limit point of
$A^f_n$. Let $\{g_{n,i}\}$ be any collection of simple functions such that for each $n$, $\<g_{n,i}\>_{i=1}^{\infty}$ is a $\PSPACE$ $L^1$-approximation of $\widetilde{A^f_n}$. Then, $\lim\limits_{n \to \infty}\lim\limits_{i \to \infty} g_{n,i}(x)=\int f d\mu$ for any $\EXP$ random $x$.
\end{theorem*}
Hence, the ideal converse to Theorem \ref{thm:pspaceergodictheorem} is the following:
\begin{theorem*}
	Given any $\EXP$ null $x$, there exists a $\PSPACE$ ergodic measure preserving transformation $T$ and $\PSPACE$ $L^1$-computable $f \in L^{1}(\Sigma^\infty,\mu)$  such that the following conditions are true:
	\begin{enumerate}
		\item $\int f d\mu$ is an $\PSPACE$-rapid limit point of $A^f_n$.
		\item There exists a collection of simple functions $\{g_{n,i}\}$ such that for each $n$, $\<g_{n,i}\>_{i=1}^{\infty}$ is a $\PSPACE$ $L^1$-approximation of $A^f_n$ but $\lim\limits_{n \to \infty}\lim\limits_{i \to \infty} g_{n,i}(x)\neq\int f d\mu$.
	\end{enumerate}
\end{theorem*}

But, we prove the following partial converse to Theorem \ref{thm:pspaceergodictheorem}.
\begin{theorem}
\label{thm:conversetopspaceergodictheorem}
	Given any $\PSPACE$ null $x$, there exists a $\PSPACE$ $L^1$-computable $f \in L^{1}(\Sigma^\infty,\mu)$ such that for any $\PSPACE$ simple measure preserving transformation, the following conditions are true:
	\begin{enumerate}
		\item For all $n\in \N$, $\lVert A_n^f -\int f d\mu \rVert_1=0$. Hence, $\int f d\mu$ is an $\PSPACE$-rapid $L^1$-limit point of $A^f_n$. 
		\item There exists a collection of simple functions $\{g_{n,i}\}$ such that for each $n$, $\<g_{n,i}\>_{i=1}^{\infty}$ is a $\PSPACE$ $L^1$-approximation of $A^f_n$ but $\lim\limits_{n \to \infty}\lim\limits_{i \to \infty} g_{n,i}(x)\neq\int f d\mu$.
	\end{enumerate}
\end{theorem}
A proof of the above theorem requires the construction in the
following lemma.
\begin{lemma}
\label{lem:finitepspacetest}
Let $\<U_n\>_{n=1}^{\infty}$ be a $\PSPACE$ test. Then there exists a sequences of sets $\<\widehat{S}_n\>_{n=1}^{\infty}$ such that for each $n \in \N$, $\widehat{S}_n \subseteq \Sigma^*$ satisfying the following conditions:
\begin{enumerate}
    \item $\mu([\widehat{S}_n]) \leq 2^{-n}$.
	\item  $\cap_{m=1}^{\infty}\cup_{n=m}^{\infty} [\widehat{S}_n]  \supseteq \cap_{n=1}^{\infty} U_n$.
	\item There exists $c \in \N$ such that for all $n$, $\sigma \in \widehat{S}_n $ implies $\lvert \sigma \rvert\leq n^c$.
	\item There exists a $\PSPACE$ machine $N$ such that
$N(\sigma,1^n) = 1$ if $\sigma \in \widehat{S}_n$ and 0 otherwise.
\end{enumerate}
\end{lemma}
\begin{proof}[Proof of Lemma \ref{lem:finitepspacetest}]
Let $\<S_n^k\>_{n,k}$ be the collection of approximating sets and $M$ be the machine computing $\<U_n\>_{n=1}^{\infty}$ as in Definition \ref{def:pspaceopensets}. Define $\mathcal{U}=\cup_{n=1}^{\infty}\cup_{k=1}^{\infty}S_n^k$.
Now, define
\begin{align*}
T_n= \bigcup\limits_{i=1}^{n+1} S_{i}^{2n+2}	
\end{align*}
 Observe that,
 \begin{align*}
 [\mathcal{U}]\setminus [T_n] \subseteq \left( \bigcup\limits_{i=1}^{n+1} U_i - [S_{i}^{2n+2}]  \right) \bigcup \left(\bigcup\limits_{i=n+2}^{\infty} U_i \right)
 \end{align*}
 Hence,
\begin{align*}
\mu([\mathcal{U}]\setminus [T_n]) &\leq 	\sum\limits_{i=1}^{n+1} \frac{1}{2^{2n+2}} + \sum\limits_{i=n+2} \frac{1}{2^{i}} \\
&\leq \frac{n+1}{2^{n+1+n+1}} + \frac{1}{2^{n+1}} \\
&\leq \frac{1}{2^n}.
\end{align*}

From the definition of $T_n$, it follows that there is a $c \in \N$ such that the length of strings in $T_n$ is upper bounded by $n^c$. Now, if $\widehat{S}_n =\{\sigma \in T_{n+1}: \forall \alpha \sqsubseteq \sigma (\alpha \not\in T_n)\}$, we have $\mu([\widehat{S}_n])\leq \mu([\mathcal{U}]\setminus [T_n]) \leq 2^{-n}$. Conditions 2 and 3 can be readily verified to be true. We now construct a $\PSPACE$ machine $N$ satisfying the condition in 4. $N$ on input $(\sigma,1^n)$ does the following:
\begin{enumerate}
	\item For each $i\in[1,(n+1)+1]$ simulate $M(\sigma,1^i,1^{2(n+1)+2})$. If all simulations result in $0$, output $0$.
	\item Else, for each $m\in [1,n]$ do the following:
	\begin{enumerate}
		\item For each $\alpha \sqsubseteq \sigma$ do the following:
		\begin{enumerate}
			\item For each $i\in[1,m+1]$ simulate $M(\alpha,1^i,1^{2m+2})$. If any of these simulations result in a $1$ then, output $0$.
		\end{enumerate}
	\end{enumerate}
	\item Output $1$.
\end{enumerate} 
$N$ can be easily verified to be a $\PSPACE$ machine. Hence, our constructions satisfy all the desired conditions.
\end{proof}

Now, we prove Theorem \ref{thm:conversetopspaceergodictheorem}.
\begin{proof}[Proof of Theorem \ref{thm:conversetopspaceergodictheorem}]
	Let $\<V_n\>_{n=1}^\infty$ be any $\PSPACE$ test such that $x \in \cap_{n=1}^{\infty} V_n$. Now, from Lemma \ref{lem:finitepspacetest}, there exists a collection of sets $\<\widehat{S}_n\>_{n=1}^{\infty}$ such that $\cap_{m=1}^{\infty}\cup_{n=m}^{\infty} [\widehat{S}_n]  \supseteq \cap_{n=1}^{\infty} V_n$. Let, 
	\begin{align*}
		U_n=\{\sigma:[\sigma] \in \widehat{S_i} \text{ for some } i \text{ such that } 2n+1 \leq i \leq 2(n+1)+1 \}
	\end{align*}
Now, let $f_n=n\chi_{U_n}$. Since, 
\begin{align*}
	\mu(U_n)\leq \sum\limits_{i=2n+1}^{2(n+1)+1} \frac{1}{2^{i}} \leq \frac{1}{2^{2n}} 
\end{align*}
it follows that
\begin{align*}
\lVert f_n \rVert_1 \leq 	\frac{n}{2^{n+n}} \leq \frac{1}{2^{n}}.
\end{align*}
Now, using the properties of $\<\widehat{S}_n\>_{n=1}^{\infty}$, it can be shown that $\<f_n\>_{n=1}^{\infty}$ is a $\PSPACE$ $L^1$-approximation of $f=0$. We construct a machine $M$ computing $\<f_n\>_{n=1}^{\infty}$. The other conditions are easily verified. Let $N$ be the machine from Lemma \ref{lem:finitepspacetest}. On input $(1^n,\sigma)$, $M$ does the following:
\begin{enumerate}
	\item If $\lvert \sigma \rvert < (2(n+1)+1)^c$ then, output $?$.
	\item Else, for each $i \in [2n+1,2(n+1)+1]$ do the following:
	\begin{enumerate}
		\item For each $\alpha \subseteq \sigma$, do the following:
		\begin{enumerate}
		\item If $N(1^i,\alpha)=1$ then, output $n$.
		\end{enumerate}
	\end{enumerate}
	\item Output $0$.
\end{enumerate}
	
	$M$ uses at most polynomial space and computes $\<f_n\>_{n=1}^{\infty}$. Now, define
	\begin{align*}
	g_{n,i}=\frac{f_i +f_i \circ T + \dots + f_i \circ T^{n-1}}{n}	
	\end{align*}
	For any fixed $n \in \N$, since $T$ is a $\PSPACE$ simple transformation, as in Lemma \ref{lem:pspacesequencelemma} it can be shown that $\<g_{n,i}\>_{i=1}^{\infty}$ is a $\PSPACE$ $L^1$-approximation of $A^f_n$. We know that there exist infinitely many $m$ such that $x \in [\widehat{S}_m]$. For any such $m$, let $i$ be the unique number such that $2i+1 \leq m \leq 2(i+1)+1$. For this $i$, $f_i(x)=i$. This shows that there exist infinitely many $i$ such that $f_i(x)=i$. Since each $f_i$ is a non-negative function, it follows that there are infinitely many $i$ with $g_{n,i} \geq i/n$. Hence, if $\lim\limits_{i \to \infty}g_{n,i}(x)$ exists, then it is equal to $\infty$. It may be the case that $\lim\limits_{i \to \infty}g_{n,i}(x)$ does not exist. In either case, $\lim\limits_{n \to \infty}\lim\limits_{i \to \infty}g_{n,i}(x)$ cannot be equal to $\int f d\mu=0$. Hence, our construction satisfies all the desired conditions.
\end{proof}

\section{An ergodic theorem for $\SUBEXP$-space randoms and its converse}
In the previous sections, we demonstrated that for $\PSPACE$
$L^1$-computable functions and $\PSPACE$ simple transformations, the
Birkhoff averages converge to the desired value over $\EXP$
randoms. However, the converse holds only over $\PSPACE$ non-randoms.
The two major reasons for this \textit{gap} are the following:
$\PSPACE$-rapid convergence necessitates exponential length cylinders
while constructing the randomness tests, and $\PSPACE$
$L^1$-computable functions are not strong enough to \textit{capture}
all $\PSPACE$ randoms. In this section, we demonstrate that for a
different notion of randomness - $\SUBEXP$-space randoms and a larger class of $L^1$-computable
functions ($\SUBEXP$-space $L^1$-computable), we can prove the ergodic theorem on the randoms and obtain
its converse on the non-randoms. Analogous to Towsner and Franklin
\cite{franklin2014randomness}, we demonstrate that the ergodic theorem
for $\PSPACE$
  simple transformations and $\SUBEXP$-space $L^1$-computable functions satisfying $\PSPACE$ rapidity, fails for
  exactly this class of non-random points. We first introduce
$\SUBEXP$-space tests and $\SUBEXP$-space randomness.

\begin{definition}[$\SUBEXP$-space sequence of open sets]  
  \label{def:subexpopensets}
  A sequence of open sets $\<U_n\>_{n=1}^\infty$ is
  a \emph{$\SUBEXP$-space sequence of open sets} if there exists a sequence
  of sets $\<S^k_n\>_{k,n \in \N}$, where $S^k_n \subseteq \Sigma^*$
  such that 
  \begin{enumerate}
  \item $U_n = \cup_{k=1}^{\infty}[S^{k}_n]$, where for any
    $m>0$, $\mu\left(U_n-\cup_{k=1}^m [S^k_n]\right)\leq
    m^{-\log(m)}$.
  \item There exists a \emph{controlling polynomial} $p$ such
    that $max\{|\sigma|:\sigma \in \cup_{k=1}^m S^k_n)\}\leq 2^{p(\log(n)+\log(m))}$.
  \item The function $g:\Sigma^* \times 1^*\times 1^* \to \{0,1\}$
    such that $g(\sigma,1^n,1^m)=1$ if $\sigma \in S^m_n$, and 0
    otherwise, 
    is decidable by a $\PSPACE$ machine.
  \end{enumerate} 
\end{definition}

\begin{definition}[$\SUBEXP$-space randomness] 
A sequence of open sets $\<U_n\>_{n=1}^\infty$ is a \emph{$\SUBEXP$-space
    test} if it is a $\SUBEXP$-space sequence of open sets and for all $n
  \in \N$, $\mu(U_n)\leq n^{-\log(n)}$.
  
  A set $A \subseteq \Sigma^\infty$ is \emph{$\SUBEXP$-space null} if there
  is a $\SUBEXP$-space test $\<U_n\>_{n=1}^\infty$ such that $A \subseteq
  \cap_{n=1}^{\infty} U_n$. A set $A\subseteq \Sigma^\infty$
  is \emph{$\SUBEXP$-space random} if $A$ is not $\SUBEXP$-space null.\footnote{It is easy to see that the set of $\SUBEXP$-space randoms is smaller than the set of $\PSPACE$ randoms. But, we do not know if any inclusion holds between $\SUBEXP$-space randoms and $\EXP$-randoms.}
\end{definition}

%

The slower decay rate of $n^{-\log(n)}=2^{-\log(n)^2}$ enables us to obtain an ergodic theorem and an exact converse in the $\SUBEXP$-space setting .The following results are useful in manipulating sums involving terms of the form $2^{-(\log(n))^k}$ for $k \geq 2$.
\begin{lemma}
\label{lem:summationlemma}
	For any $k,m \in \N$, $\sum_{i=m+1}^{\infty} \frac{1}{n^{k+1}} \leq \frac{1}{m^k}$
\end{lemma}
\begin{proof}
	Let $k=1$. Then, for any $m > 0$,
	\begin{align*}
		\sum\limits_{i=m+1}^{\infty} \frac{1}{n^{2}} &< \frac{1}{m(m+1)} + \frac{1}{(m+1)(m+2)} + \dots \\
		&= \frac{1}{m}-\frac{1}{m+1}+\frac{1}{m+1}-\frac{1}{m+2}+\dots \\
		&= \frac{1}{m}	
	\end{align*}
	For $k>1$,
	\begin{align*}
		\sum\limits_{i=m+1}^{\infty} \frac{1}{n^{k+1}} < \frac{1}{m^{k-1}} \sum\limits_{i=m+1}^{\infty} \frac{1}{n^{2}} 
	\end{align*}
	The proof now follows by applying the result when $k=1$ to bound the summation on the right with $m^{-1}$.
\end{proof}

\begin{lemma}
\label{lem:mpolylogsummation}
For any $m\in \N$, $\sum_{n=2(2m^2+1)}^{\infty} \frac{n}{n^{\log(n)}} \leq \frac{1}{m^{\log(m)}}$.
\end{lemma}
\begin{proof}
\begin{align*}
	\sum\limits_{n=2(2m^2+1)}^{\infty} \frac{n}{n^{\log(n)}} &\leq  \sum\limits_{n=2(2m^2+1)}^{\infty} \frac{1}{n^{ \log(n)-1 }}\\
	&\leq \sum\limits_{n=2m^2+1}^{\infty} \frac{1}{n^{ \log(2m^2) }}\\
	&\leq \sum\limits_{n=2m^2+1}^{\infty} \frac{1}{n^{ \lfloor \log(m^2) \rfloor+1}}\\
	&\leq \frac{1}{(2m)^{2(\lfloor \log(m^2) \rfloor)}}\\
	&\leq \frac{1}{m^{2\frac{\log(m^2)}{2} }}\\
	&\leq \frac{1}{m^{\log(m)}}
\end{align*}	

The fourth inequality above follows from Lemma \ref{lem:summationlemma}.
\end{proof}

A similar inequality can be trivially seen to be true on replacing
$n/n^{\log(n)}$ with $1/n^{\log(n)}$. Now, we introduce the Solovay analogue of $\SUBEXP$-space randomness
and prove that these notions are analogous.
\begin{definition}[$\SUBEXP$-space Solovay test]
  A sequence of open sets $\<U_n\>_{n=1}^\infty$ is
  a \emph{$\SUBEXP$-space Solovay test} if it is a $\SUBEXP$-space sequence of
  open sets and there exists a polynomial $p$ such that $\forall m \geq
          0$, $\sum_{n=p(m)+1}^{\infty}\mu(U_n) \leq
          \frac{1}{m^{\log(m)}}$\footnote{This implies that $\sum_{n=1}^{\infty}\mu(U_n) < \infty $}. A set $A \subseteq \Sigma^\infty$ is \emph{$\SUBEXP$-space Solovay null}
  if there exists a $\SUBEXP$-space Solovay test $\<U_n\>_{n=1}^\infty$ such
  that $A \subseteq
  \cap_{i=1}^{\infty}\cup_{n=i}^{\infty}
  U_n$. $A\subseteq \Sigma^\infty$ is \emph{$\SUBEXP$-space Solovay random}
  if $A$ is not $\SUBEXP$-space Solovay null.  
\end{definition}

\begin{lemma}
  \label{lem:subexpsolovay_equivalence}
A set $A \subseteq \Sigma^\infty$ is $\SUBEXP$-space null if and only if $A$
is $\SUBEXP$-space Solovay null. 	 
\end{lemma}
\begin{proof}
Using Lemma \ref{lem:mpolylogsummation}, it is easy to see that if $A$ is $\SUBEXP$-space null then $A$ is $\SUBEXP$
Solovay null. Conversely, let $A$ be $\SUBEXP$-space Solovay null and let
$\<U_n\>_{n=1}^{\infty}$ be any Solovay test which witnesses this
fact. Let $V_{n}=\cup_{i=p(n)+1}^{\infty} U_n$. We show that
$\<V_n\>_{n=1}^\infty$ is a $\SUBEXP$-space test. Let $\<S_n^k\>_{n,k \in
  \N}$ be any sequence of sets approximating $\<U_n\>_{n=1}^{\infty}$
as in definition \ref{def:subexpopensets} such that
$\<S_n^k\>_{k=1}^{\infty}$ is increasing for each $n$. We define
a sequence of sets $\<T_n^k\>_{n,k \in \N}$ approximating $V_n$ as
follows.

Let $r(n,k)=\max \{2((2(2k^2+1)+1)^2+1),p(n)+1\}$. Define
\begin{align*}
T_n^k = \bigcup\limits_{i=p(n)+1}^{r(n,k)} S_i^{(r(n,k)-p(n))2(2k^2+1)}.	
\end{align*}

We can easily verify conditions 1 and 3 in definition
\ref{def:subexpopensets}. From the definition of $T_n^k$, it can be verified that
\begin{align*}
	\mu(V_n-[T_n^k]) &\leq \sum\limits_{i=p(n)+1}^{r(n,k)}2^{-\log(r(n,k)-p(n))^2-\log(2(2k^2+1))^2}+\sum\limits_{n=r(n,k)+1}^{
\infty}2^{-\log(n)^2} \\
	&\leq \sum\limits_{i=p(n)+1}^{r(n,k)}2^{-\log(r(n,k)-p(n))^2-\log(2(2k^2+1))^2}+\sum\limits_{n=2((2(2k^2+1)+1)^2+1)}^{
\infty}2^{-\log(n)^2} \\
&\leq \frac{r(n,k)-p(n)}{2^{\log(r(n,k)-p(n))^2}}\frac{1}{2^{\log(2(2k^2+1))^2}} + \frac{1}{2^{\log(2(2k^2+1)+1)^2}} \\
&< \frac{1}{2^{\log(2(2k^2+1))^2}} + \frac{1}{2^{\log(2(2k^2+1)+1)^2}} \\
&< \frac{1}{2^{\log(k)^2}}.
\end{align*}
The third inequality and the last inequality above follows from Lemma \ref{lem:mpolylogsummation}. Using the machine $M$ and controlling polynomial $p$ witnessing that
$\<U_n\>_{n=1}^{\infty}$ is a $\PSPACE$ sequence of open sets, we can
construct the corresponding machines for $\<V_n\>_{n=1}^{\infty}$ in the following way. Machine $N$ on input $(\sigma,1^n,1^k)$ does the following:
\begin{enumerate}
	\item For each $i \in [p(n)+1,r(n,k)]$ do the following:
	\begin{enumerate}
	\item Output $1$ if $M(\sigma,1^i,1^{(r(n,k)-p(n))2(2k^2+1)})=1$. 
	\end{enumerate}
	\item Output $0$ if none of the above computations results in $1$.
\end{enumerate}
It is straightforward to verify that $N$ is a $\PSPACE$ machine.
\end{proof}

Now, we define $\SUBEXP$-space analogues of concepts from Section \ref{sec:pspacel1computability}.
\begin{definition}[$\SUBEXP$-space sequence of simple functions]
\label{def:subexpsequenceofsimplefunctions}
	A sequence of simple functions $\<f_n\>_{n=1}^{\infty}$ where
        each $f_n : \Sigma^\infty \to \Q$ is a \emph{$\SUBEXP$-space
          sequence of simple functions} if
	\begin{enumerate}
	\item There is a \emph{controlling polynomial} $p$ such that
          for each $n$, there exist $k(n)\in \N$, $\{d_1,d_2 \dots,
          d_{k(n)}\}\subseteq \Q$ and $\{\sigma_1,\sigma_2 \dots
          \sigma_{k(n)}\} \subseteq \Sigma^{2^{p(\log(n))}}$ such that
          $f_n = \sum_{i=1}^{k(n)}d_i \chi_{\sigma_i}$,
          where $\chi_{\sigma_i}$ is the characteristic function of
          the cylinder $[\sigma_i]$.
	\item There is a $\PSPACE$ machine $M$ such that for each $n
          \in \N, \sigma \in \Sigma^{*}$
	\begin{align*}
		M(1^n,\sigma)=\begin{cases}
                f_n(\sigma0^\infty) & \text{if } |\sigma| \geq 2^{p(\log(n))}\\
                ? & \text{otherwise.}
                \end{cases}
	\end{align*}
	\end{enumerate}
\end{definition} 

\begin{definition}[$\SUBEXP$-space $L^1$-computable functions]
\label{def:subexpl1computablefunction}
A function $f \in L^1(\Sigma^\infty,\mu)$ is $\SUBEXP$-space $L^1$-computable
if there exists a $\SUBEXP$-space sequence of simple functions
$\<f_n\>_{n=1}^{\infty}$ such that for every $n \in \N$,
$\lVert f-f_n \rVert \leq n^{-\log(n)}$. The sequence
$\<f_n\>_{n=1}^{\infty}$ is called a \emph{$\SUBEXP$-space
$L^1$-approximation} of $f$.
\end{definition}

We require the following equivalent definitions of $\PSPACE$-rapid convergence notions for working in the setting of $\SUBEXP$-space randomness.
\begin{lemma}
\label{lem:pspacerapidequivalence}
	A real number $a$ is a $\PSPACE$-rapid limit point of the real
  number sequence $\<a_n\>_{n=1}^{\infty}$ if and only if there exists a
  polynomial $p$ such that for all $m \in \N$, $ \exists k \leq 2^{p(\log(m))}
  \text{ such that }|a_{k}-a| \leq m^{-\log(m)}$.
\end{lemma}
\begin{proof}
We prove the forward direction first. If $a$ is a $\PSPACE$-rapid limit point of $\<a_n\>_{n=1}^{\infty}$, there exists a polynomial $q$ such that for all $l \in \N$, $ \exists k \leq 2^{q(l)}
  \text{ such that }\lvert a_{l}-a\rvert \leq 2^{-m}$. Substituting $l=\lceil \log(m) \rceil$, we get that for all $m \in \N$, $ \exists k \leq 2^{q(\lceil \log(m) \rceil)} \leq 2^{q(\log(m)+1)}
  \text{ such that }\lvert a_{l}-a\rvert \leq 2^{-\lceil \log(m)\rceil} \leq 2^{-\log(m)}$. Conversely, assume that there exists a
  polynomial $p$ such that for all $m \in \N$, $ \exists k \leq 2^{p(\log(m))}
  \text{ such that }|a_{k}-a| \leq 2^{-\log(m)^2}$.  Substituting $m=2^{\lceil l^\frac{1}{2} \rceil}$, we get that for all $m \in \N$, $ \exists k \leq 2^{q(\lceil l^\frac{1}{2} \rceil)} \leq 2^{q(l)}
  \text{ such that }\lvert a_{l}-a\rvert \leq 2^{-\lceil l^\frac{1}{2} \rceil^2}\leq 2^{-l}$. Hence, $a$ is a $\PSPACE$-rapid limit point of $\<a_n\>_{n=1}^{\infty}$.
\end{proof}

\begin{lemma}
	A sequence of measurable functions $\<f_n\>_{n=1}^{\infty}$ is
  $\PSPACE$-rapid almost everywhere convergent to a measurable
  function $f$ if and only if there exists a
  polynomial $p$ such that for all
  $m_1$ and $m_2$,
  \begin{align*}
    \mu\left( \left\{ x : \sup\limits_{n \geq 2^{p(\log(m_1)+\log(m_2))}}
    |f_n(x)-f(x)| \geq \frac{1}{m_1^{\log(m_1)}}\right\} \right) \leq
    \frac{1}{m_2^{\log(m_2)}}.
  \end{align*}

\end{lemma}
The same technique used in the proof of Lemma \ref{lem:pspacerapidequivalence} can be used to prove this claim.

Before addressing the main result, let us define $\SUBEXP$-space ergodicity.
\begin{definition}[$\SUBEXP$-space ergodic transformations]
	A measurable function $T:(\Sigma^\infty,\mu) \to (\Sigma^\infty,\mu)$ is \emph{$\SUBEXP$-space ergodic} if $T$ is a $\PSPACE$ simple transformation such that for any $\SUBEXP$-space $L^1$-computable $f\in L^1 (\Sigma^\infty,\mu)$, $\int f d\mu$ is a $\PSPACE$-rapid $L^1$ limit point of $A_n^{f,T}$.
\end{definition}
Now, we prove $\SUBEXP$ analogues of the auxiliary lemmas from Section \ref{sec:pspaceergodictheorem}.
\begin{lemma}
\label{lem:subexpconvergencelemma}
Let $\<f_n\>_{n=1}^{\infty}$ be a $\SUBEXP$-space sequence of simple
functions which converges $\PSPACE$-rapid almost everywhere to $f \in L^1
(\Sigma^\infty,\mu)$. Then,
\begin{enumerate}
	\item $\lim\limits_{n \to \infty}f_n(x)$ exists for all $\SUBEXP$-space random $x$.
	\item Given a $\SUBEXP$-space sequence of simple functions $\<g_n\>_{n=1}^{\infty}$ which is $\PSPACE$-rapid almost everywhere convergent to $f$, $\lim\limits_{n \to \infty}g_n(x)=\lim\limits_{n \to \infty}f_n(x)$ for all $\SUBEXP$-space random $x$. 
\end{enumerate}	
\end{lemma}
\begin{proof}
We initially show 1. For each $k \geq 0$, since $\<f_n\>_{n=1}^{\infty}$ is $\PSPACE$-rapid almost everywhere convergent to $f$, there exists $c \in \N$ such that
\begin{align*}
\mu \left( \left\{x : \sup\limits_{n \geq 2^{\lceil \log(k)^c \rceil}} |f_n(x)-f(x)| \geq \frac{1}{2^{\log(2k)^2}} \right\} \right) \leq \frac{1}{2^{\log(2k)^2}}.
\end{align*}
Since $(\log(2k))^2=(\log(k)+1)^2\geq \log(k)^2+1$, we get
\begin{align*}
\mu \left( \left\{x : \sup\limits_{n \geq 2^{\lceil \log(k)^c \rceil}} |f_n(x)-f(x)| \geq \frac{1}{2^{\log(k)^2+1}} \right\} \right) \leq \frac{1}{2^{\log(k)^2+1}}.
\end{align*}
	It is easy to verify that
\begin{align*}
\mu \left( \left\{x : \sup\limits_{n \geq 2^{\lceil \log(k)^c \rceil }} |f_n(x)-f_{2^{\lceil \log(k)^c \rceil}}(x)| \geq \frac{1}{2^{\log(k)^2}} \right\} \right) \leq \frac{1}{2^{\log(k)^2}}.
\end{align*}	
Define
\begin{align*}
U_k = \left\{x:\max\limits_{2^{\lceil \log(k)^c \rceil} \leq n \leq 2^{\lceil \log(k+1)^c \rceil}} |f_n(x)-f_{2^{\lceil \log(k)^c \rceil}}(x)|> \frac{1}{2^{\lfloor \log(k)^2 \rfloor}}\right\}.
\end{align*}
Observe that
\begin{align*}
\mu(U_k) \leq \mu \left( \left\{x : \sup\limits_{n \geq 2^{\lceil \log(k)^c \rceil}}
|f_n(x)-f_{2^{\lceil \log(k)^c \rceil}}(x)| \geq \frac{1}{2^{\log(k)^2}} \right\} \right) \leq
\frac{1}{2^{\log(k)^2}}.	 
\end{align*}
Let $q$ be the controlling polynomial and let $M$ be the $\PSPACE$
machine witnessing the fact that $\<f_n\>_{n=1}^{\infty}$ is a
$\SUBEXP$-space sequence of simple functions. $U_k$ is hence a union of
cylinders of length at most 
\begin{align*}
2^{q(\log(2^{\lceil \log(k+1)^c \rceil}))}	= 2^{q(\lceil \log(k+1)^c \rceil)} 
\end{align*}
which is upper bounded by $2^{\log(k)^d}$ for some $d \in \N$. The fact that functions of the form $2^{\log(n)^i}$ are closed under composition enables us to obtain convergence on $\SUBEXP$-space randoms instead of $\EXP$ randoms as in Lemma \ref{lem:convergencelemma}. The machine $M$ can be used to
construct a machine $N$ that on input $(\sigma,1^k)$ outputs $1$ if $[\sigma]
\subseteq U_k$ and outputs $0$ otherwise. Let the  machine $N$ on input $(\sigma,1^k)$ do the following:

\begin{enumerate}
	\item If $\lvert \sigma \rvert < 2^{q(\lceil \log(k+1)^c \rceil)}$ then, output $0$.
	\item Compute $f_{2^{\lceil \log(k)^c \rceil}}(\sigma0^\infty)$ by running $M(1^{2^{\lceil \log(k)^c \rceil}}, \sigma) $ and store the result. 
	\item For each $n \in [2^{\lceil \log(k)^c \rceil},2^{\lceil \log(k+1)^c \rceil}]$ do the following:
	\begin{enumerate}
	\item Compute $f_{n}(\sigma0^\infty)$ by running $M(1^{n}, \sigma) $ and store the result.
	\item Check if $|f_n(\sigma0^\infty)-f_{2^{\lceil \log(k)^c \rceil}}(\sigma0^\infty)| \geq \frac{1}{2^{\lfloor \log(k)^2 \rfloor}}$. If so, output $1$.
	\end{enumerate}
	\item Output $0$
\end{enumerate}
Since $N$ rejects any $\sigma$ with length less than $2^{q(\lceil \log(k+1)^c \rceil)}$, the simulation of $M(1^{2^{\lceil \log(k)^c \rceil}}, \sigma)$ is always a polynomial space operation. Hence, $N$ witnesses the fact that $\<U_k\>_{k=1}^{\infty}$ is an $\EXP$ sequence of open sets.

Define
\begin{align*}
V_k = \bigcup\limits_{i=2(2k^2+1)}^{\infty} U_i.	
\end{align*}
It follows from Lemma \ref{lem:mpolylogsummation} that $\mu(V_k)\leq 2^{-\log(k)^2}$ and from the above observations it can be verified that $\<V_k\>_{k=1}^{\infty}$ is an $\SUBEXP$-space test.

If $x \in \Sigma^\infty$ is an $\SUBEXP$-space random then $x$ is in at most
finitely many $V_k$ and hence in only finitely many $U_k$. Hence, for
some $k$ and for all $n \geq 2^{\lceil \log(k)^c \rceil}$ using Lemma \ref{lem:summationlemma} we have, 
\begin{align*}
\left|f_n(x)-f_{2^{\lceil \log(k)^c \rceil}}(x)\right|	 \leq \sum\limits_{j=k}^{\infty}
\frac{1}{2^{\log(j)^2}} \leq \sum\limits_{j=k}^{\infty}
\frac{1}{j^{\lfloor \log(k) \rfloor}} \leq \frac{1}{(k-1)^{(\lfloor \log(k) \rfloor-1)}}.
\end{align*}
This shows that $f_n(x)$ is a Cauchy sequence. This completes the
proof of 1. 

Given $\<g_n\>_{n=1}^{\infty}$ which is $\PSPACE$-rapid almost everywhere
convergent to $f$, the interleaved sequence $f_1,g_1,f_2,g_2,f_3,g_3
\dots$ can be easily verified to be $\PSPACE$-rapid almost everywhere
convergent to $f$. 2 now follows directly from 1.  
\end{proof}
The following immediately follows from the above lemma.  
\begin{corollary}
\label{cor:subexpconvergencelemmacorollary1}
Let $f \in L^1 (\Sigma^\infty,\mu)$ be a $\SUBEXP$-space $L^1$-computable
function with an $L^1$ approximating $\SUBEXP$-space sequence of simple functions
$\<f_n\>_{n=1}^{\infty}$. Then,
\begin{enumerate}
\item $\lim\limits_{n \to \infty}f_n(x)$ exists for all $\SUBEXP$
  random $x$. 
\item Given a $\SUBEXP$-space sequence of simple functions
  $\<g_n\>_{n=1}^{\infty}$ $L^1$ approximating $f$, $\lim\limits_{n \to
  \infty}g_n(x)=\lim\limits_{n \to \infty}f_n(x)$ for all $\SUBEXP$
  random $x$. 
\end{enumerate}	
\end{corollary}

\begin{proof}
	For
any $m_1,m_2 \geq 0$,
\begin{align*}
\mu \left( \left\{x : \sup\limits_{n \geq 2(2(m_1m_2)^2+1)} |f_n(x)-f(x)|
\geq \frac{1}{2^{\log(m_1)^2}} \right\} \right) &\leq \sum\limits_{n=2(2(m_1m_2)^2+1)}^{\infty} \mu \left( \left\{x : |f_n(x)-f(x)| \geq
\frac{1}{2^{\log(m_1)^2}} \right\} \right)
\\ &\leq  \sum\limits_{n=2(2(m_1m_2)^2+1)}^{\infty} \lVert f_n-f\rVert_1 2^{\log(m_1)^2} 
\\ &\leq 2^{\log(m_1)^2}\sum\limits_{n=2(2(m_1m_2)^2+1)}^{\infty} \frac{1}{2^{\log(n)^2}} 
\\ &\leq \frac{2^{\log(m_1)^2}}{2^{\log(m_1m_2)^2}}
\\ &\leq \frac{2^{\log(m_1)^2}}{2^{\log(m_1)^2+\log(m_2)^2}}
\\ &=\  \frac{1}{2^{\log(m_2)^2}}.
\end{align*}
Hence, $\<f_n\>_{n=1}^{\infty}$ is $\PSPACE$-rapid almost everywhere
convergent to $f$. The claim now follows from Lemma \ref{lem:subexpconvergencelemma}.
\end{proof}

The following are $\SUBEXP$-space analogues of Lemma \ref{lem:pspacesequencelemma} and Lemma \ref{lem:pspacefunctionintegral}.

\begin{lemma}
\label{lem:subexpfunctionintegral}
Let $f$ be a $\SUBEXP$-space $L^1$-computable function. Let $I_f:\Sigma^\infty \to \Sigma^\infty$ be the constant function taking the value $\int f d\mu$ over all $x \in \Sigma^\infty$. Then, $I_f$ is $\SUBEXP$-space $L^1$-computable and $\widetilde{I_f}(x)=\int f d\mu$ for all $\SUBEXP$-space random $x$.
\end{lemma}
\begin{proof}
Let $\<f_n\>_{n=1}^{\infty}$ be a $\SUBEXP$-space sequence of simple functions, $M$ be a $\PSPACE$ machine and $p$ be a controlling polynomial witnessing the fact that $f$ is $\SUBEXP$-space $L^1$-computable. We construct a $\SUBEXP$-space sequence of simple functions $\<f'_n\>_{n=1}^{\infty}$ where each $f'_n$ is the constant function taking the value $\int f_n d\mu$. Since $\lVert f'_n - I_f \rVert_1 = \lvert \int f_n d\mu - \int f d\mu \rvert \leq \lVert f_n -f \rVert_1 \leq  n^{-\log(n)}$, it follows that $I_f$ is $\SUBEXP$ $L^1$-computable. Now, from part 2 of Lemma \ref{lem:subexpconvergencelemma}, we get that $\widetilde{I_f}(x)=\lim\limits_{n \to \infty}f'_n(x)=\lim\limits_{n \to \infty}\int f_n d\mu=\int f d\mu$ for all $\SUBEXP$-space random $x$.

On input $(1^n,\sigma)$, let machine $N$ do the following:
\begin{enumerate}
	\item If $|\sigma| < 2^{\lceil p (\log(n)) \rceil}$ output $?$. Else, do the following:
	\item Let $\mathrm{Sum}=0$.
	\item For each $\alpha \in \Sigma^{2^{\lceil p (\log(n)) \rceil}}$ do the following:
	\begin{enumerate}
		\item Run $M(1^n,\alpha)$ and add the result to $\mathrm{Sum}$.
	\end{enumerate}
	\item Output $\mathrm{Sum}/2^{2^{\lceil p (\log(n)) \rceil}}$.
\end{enumerate}
Let $t$ be a polynomial upper bound for the space complexity of $M$. Then, the result of $M(1^n,\alpha)$ is always upper bounded by $2^{t(n+2^{\lceil p (\log(n)) \rceil})}$ and representable in $t(n+2^{\lceil p (\log(n)) \rceil})$ space. The sum of at most $2^{2^{\lceil p (\log(n)) \rceil}}$ many such numbers is upper bounded by $2^{t(n+2^{\lceil p (\log(n)) \rceil})+2^{\lceil p (\log(n)) \rceil}}$ and representable in $t(n+2^{\lceil p (\log(n)) \rceil})+2^{\lceil p (\log(n)) \rceil}$ space. Since $N$ rejects any $\sigma$ with $|\sigma| < 2^{\lceil p (\log(n)) \rceil}$, calculating the running sum can be done in polynomial space.  Dividing the running sum by $2^{2^{\lceil p (\log(n)) \rceil}}$ can also be done in polynomial space due to the same reason. Hence, $N$ is a $\PSPACE$ machine computing the sequence $\<f'_n\>_{n=1}^{\infty}$ where each $f'_n=\int f_n d\mu$.
\end{proof}

\begin{lemma}
\label{lem:subexpsequencelemma}
Let $f$ be a $\SUBEXP$-space $L^1$-computable function with an $L^1$ approximating
$\SUBEXP$-space sequence of simple functions $\<f_n\>_{n=1}^{\infty}$. Let
$T$ be a $\PSPACE$ simple transformation and $p$ be a
polynomial. Then, $\<A_n^{f_{p(n)}}\>_{n=1}^{\infty}$ is a $\SUBEXP$-space
sequence of simple functions.
\end{lemma}

\begin{proof}
	Let $q$ be a controlling polynomial and $M_f$ be a machine
        witnessing the fact that $\<f_n\>_{n=1}^{\infty}$ is a
        $\SUBEXP$-space sequence of simple functions. Let $c_T$ be a controlling
        constant witnessing the fact that $T$
        is a $\PSPACE$ simple transformation. For any $n \geq 1$, we
        have 
	\begin{align*}
	A^{f_{p(n)}}_n  = \frac{f_{p(n)}+f_{p(n)}\circ T+f_{p(n)}\circ
          T^2+\dots f_{p(n)}\circ T^n}{n}.
	\end{align*}
	The functions $\<f_{p(n)}\circ T^i\>_{i=1}^{n}$ are simple
        functions defined on cylinders of length at most $2^{q(\log(p(n)))}+c_T
        n$. Hence, $r(n)=2^{q(\lceil\log(p(n))\rceil)}+c_T n$, which can be upper bounded by $2^{\log(n)^d}$ for some $d \in \N$ is a
        controlling function for the sequence of functions
        $\<A^{f_{p(n)}}_n\>_{n=1}^{\infty}$ as in condition 1 of
        Definition \ref{def:subexpsequenceofsimplefunctions}. Now, we verify condition 2 of Definition
        \ref{def:subexpsequenceofsimplefunctions}. Let $M$ be the machine from Lemma \ref{lem:pspacetransformationcomputation}. We construct a
        machine $N$ such that for each $n \in \N$ and $\sigma \in
        \Sigma^*$, 
	\begin{align*}
		N(1^n,\sigma)=\begin{cases}
                A^{f_{p(n)}}_n(\sigma0^\infty) & \text{if } |\sigma| \geq r(n)\\
                ? & \text{otherwise.}
                \end{cases}
	\end{align*}
		On input $(1^n,\alpha)$ if $|\alpha|<r(n)=$ then $N$ outputs $?$ else it operates as follows:
		\begin{enumerate}
			\item Let $\mathrm{Sum}=0$
			\item For each $i \in [1,n]$, do the following:
			\begin{enumerate}
				\item For each string $\sigma $ of length $2^{q(\lceil\log(p(n))\rceil)}$, do the following:
				\begin{enumerate}
				\item If $M(1^i,\sigma,\alpha)=1$, then let $\mathrm{Sum}=\mathrm{Sum}+M_f(1^{p(n)},\sigma)$.
				\end{enumerate}
			 \end{enumerate}
			 \item Output $\mathrm{Sum}/n$.
		\end{enumerate}
		
		 If $t_1$ is a polynomial upper bound for the space complexity of $M$ and $t_2$ is a polynomial upper bound for the space complexity of $M_f$ then, each $f_{p(n)}\circ T^i$ can be computed in $O(t_1(2^{q(\lceil\log(p(n))\rceil)}+c_T n+n)+t_2(2^{q(\lceil\log(p(n))\rceil)}+p(n)))$ space for any $i \leq n$. Since every $\alpha$ such that $|\alpha|<r(n)=2^{q(\lceil\log(p(n))\rceil)}+c_T n$ is rejected by $N$, the computations of $f_{p(n)}\circ T^i$  is done in polynomial space. The results of computations of $f_{p(n)}\circ T^i$ for $i\in [1,n]$ can be added up and divided by $n$ in polynomial space. $N$ outputs the result of this computation. Since $N$ is a $\PSPACE$ machine, the proof is complete.
\end{proof} 
Now, we proceed onto proving the $\SUBEXP$-space ergodic theorem. 

\begin{theorem}
\label{thm:subexpergodictheorem}
Let $T:(\Sigma^\infty,\mathcal{B}(\Sigma^\infty),\mu) \to (\Sigma^\infty,\mathcal{B}(\Sigma^\infty),\mu)$ be a $\SUBEXP$-space ergodic measure preserving
transformation. Then, for any $\SUBEXP$-space $L^1$-computable $f$, $\lim\limits_{n \to \infty}\widetilde{A^f_n}=\int f
d\mu$ on $\SUBEXP$-space randoms.
\end{theorem}
\begin{proof}
Let $\<f_m\>_{m=1}^{\infty}$ be any $\SUBEXP$-space sequence of simple
functions $L^1$ approximating $f$. We initially approximate $A_n^f$ with a
$\SUBEXP$-space sequence of simple functions $\<g_n\>_{n=1}^{\infty}$ which
converges to $\int f d\mu$ on $\SUBEXP$-space randoms. Then we show that
$\widetilde{A}^f_n$ has the same limit as $g_n$ on $\SUBEXP$-space randoms.
	
For each $n$, it is easy to verify that $\<A^{f_m}_n\>_{m=1}^{\infty}$ is a $\SUBEXP$
sequence of simple functions $L^1$ approximating $A^f_n$ with the same rate
of convergence. Using techniques similar to those in Lemma
\ref{lem:subexpconvergencelemma} and Corollary
\ref{cor:subexpconvergencelemmacorollary1}, we can obtain a polynomial $p$
such that
\begin{align*}
  \mu \left( \left\{x : \sup\limits_{m \geq p(n+i)}
  |A^{f_m}_n(x)-A^{f_{p(n+i)}}_n(x)| \geq \frac{1}{2^{\log(n+i)^2}} \right\}
  \right) \leq \frac{1}{2^{\log(n+i)^2}}. 
\end{align*}
For every $n>0$, let $g_n=A^{f_{p(n)}}_n$. We initially show that $\<g_n\>_{n=1}^{\infty}$ converges to
$\int f d\mu$ on $\SUBEXP$-space randoms. Let $m_1,m_2 \geq 0$. From
Theorem \ref{thm:aeconvergence}, $A^f_n$ is $\PSPACE$-rapid almost
everywhere convergent to $\int f d\mu$. Hence there is a $d \in \N$ such that 
\begin{align*}
  \mu \left( \left\{x : \sup\limits_{n \geq 2^{(\log(m_1)+\log(m_2))^d}}
  |A^f_n(x)-\int f d\mu| \geq \frac{1}{2^{\log(m_1)^2+1}} \right\} \right)
  \leq \frac{1}{2^{\log(m_2)^2+1}}. 
\end{align*}
Let $N(m_1,m_2)=\max\{2((2m_2)^2+1),2^{\log(m_1+m_2)^d}\}$. Using Lemma \ref{lem:mpolylogsummation},
\begin{align*}
  \sum\limits_{n \geq N(m_1,m_2)} \frac{1}{2^{\log(n)^2}} \leq \sum\limits_{n \geq 2((2m_2)^2+1)} \frac{1}{2^{\log(n)^2}} \leq \frac{1}{2^{\log(2m_2)^2}}
  \end{align*}
Now, we have  
\begin{align*}
  \mu \left( \left\{x : \sup\limits_{n \geq N(m_1,m_2)} |g_n-\int f
  d\mu| > \frac{1}{2^{m_1}} \right\} \right) &\leq \sum\limits_{n \geq
    N(m_1,m_2)} \mu \left( \left\{x :  |g_n-A^f_n(x)| >
  \frac{1}{2^{\log(m_1)^2+1}}  \right\} \right)\\ 
  &+ \mu \left( \left\{x : \sup\limits_{n \geq 2^{\log(m_1+m_2)^d}}
  |A^f_n(x)-\int f d\mu| \geq \frac{1}{2^{\log(m_1)^2+1}} \right\} \right) \\ 
  &\leq \sum\limits_{n \geq N(m_1,m_2)} \frac{1}{2^{\log(n)^2}} +
  \frac{1}{2^{\log(m_2)^2+1}} \\ 
   &\leq \frac{1}{2^{\log(2m_2)^2}} + \frac{1}{2^{\log(m_2)^2+1}}\\
  &\leq \frac{1}{2^{\log(m_2)^2+1}} + \frac{1}{2^{\log(m_2)^2+1}}\\
  &\leq \frac{1}{2^{\log(m_2)^2}}.
\end{align*}
Hence,
$g_n$ is $\PSPACE$-rapid almost everywhere convergent to $\int f
d\mu$. From Lemma \ref{lem:subexpsequencelemma} it follows that
$\<g_n\>_{n=1}^{\infty}=\<A^{f_{p(n)}}_n\>_{n=1}^{\infty}$ is a
$\SUBEXP$-space sequence of simple functions (in parameter $n$). Let $I_f:\Sigma^\infty \to \Sigma^\infty$ be the constant function taking the value $\int f d\mu$ over all $x \in \Sigma^\infty$. From the above
observations and Lemma \ref{lem:subexpconvergencelemma} we get
that $\lim\limits_{n \to \infty}g_n (x)=\widetilde{I_f}(x)$ for any $x$ which
is $\SUBEXP$-space random. Now, from Lemma \ref{lem:subexpfunctionintegral}, we get that $\lim\limits_{n \to \infty}g_n (x)=\int f d\mu$ for any $x$ which
is $\SUBEXP$-space random.

We now show that $\lim\limits_{n \to \infty} \widetilde{A}^f_n =
\lim\limits_{n \to \infty} g_n$ on $\SUBEXP$-space randoms. Define 
\begin{align*}
U_{n,i} = \left\{x:\max\limits_{p(n+i) \leq m \leq p(n+i+1)}
|A^{f_m}_n(x)-A^{f_{p(n+i)}}_n(x)| \geq \frac{1}{2^{\log(n+i)^2}}\right\}. 	
\end{align*}
We already know $\mu(U_{n,i})\leq 2^{-\log(n+i)^2}$. $U_{n,i}$ can be shown to be polynomial space
approximable in parameters $n$ and $i$ in the following sense. There exists a sequence of sets of strings $\<S_{n,i}\>_{i,n \in \N}$ and polynomial $p$ satisfying the following conditions:
\begin{enumerate}
  \item $U_{n,i} = [S_{n,i}]$.
   \item There exists a \emph{controlling polynomial} $r$ such
    that $max\{|\sigma|:\sigma \in  S_{n,i}\}\leq 2^{r(\log(n)+\log(i))}$.
  \item The function $g:\Sigma^* \times 1^*\times 1^* \to \{0,1\}$
    such that
    \begin{align*}
      g(\sigma,1^n,1^i)=
      \begin{cases}
	1 &\text{if }\sigma \in S_{n,i}\\ 0 &\text{otherwise,}
      \end{cases}	
    \end{align*}
    is decidable by a $\PSPACE$ machine.
\end{enumerate}
The above claims can be established by using techniques similar to those in Lemma \ref{lem:subexpsequencelemma} and Lemma \ref{lem:subexpconvergencelemma}. We show the construction of a machine $N$ computing the function $g$ above. Let $M_f$ be a computing machine and let $q$ be a controlling polynomial for $\<f_n\>_{n=1}^{\infty}$. Let $c$ be a controlling constant for $T$. Let $M'$ be the machine from Lemma \ref{lem:pspacetransformationcomputation}. Machine $N$ on input $(\sigma,1^n,1^i)$ does the following:
\begin{enumerate}
	\item If $|\sigma| > 2^{q(\lceil\log(p(n+i+1))\rceil)}+cn$, then output $0$.
	\item Compute $A_n^{f_{p(n+i)}}(\sigma0^\infty)$ as in Lemma \ref{lem:subexpsequencelemma} by using $M_f$ and $M'$ and store the result.
	\item For each $m \in [p(n+i),p(n+i+1)]$ do the following:
	\begin{enumerate}
	\item Compute $A_n^{f_{m}}(\sigma0^\infty)$ as in Lemma \ref{lem:subexpsequencelemma} by using $M_f$ and $M'$ and store the result.
	\item Check if $|A_n^{f_{m}}(\sigma0^\infty)-A_n^{f_{p(n+i)}}(\sigma0^\infty)| \geq \frac{1}{2^{\log(n+i)^2}}$. If so, output $1$.
	\end{enumerate}
	\item Output $0$.
\end{enumerate}
It can be easily verified that $N$ is a $\PSPACE$ machine. The second condition follows from the fact that $|\sigma| \leq 2^{q(\lceil\log(p(n+i+1))\rceil)}+cn$ for any $\sigma \in S_{n,i}$. Now, define 
\begin{align*}
V_m = \bigcup\limits_{\substack{n,i \geq 0 \\ n+i=m}}	 U_{n,i}.
\end{align*}
Now, we show that $\<V_m\>_{m=1}^{\infty}$ is a $\SUBEXP$
Solovay test. Note that 
\begin{align*}
\mu(V_m) \leq \frac{m}{2^{\log(m)^2}}.	
\end{align*}

Since each $V_m$ is a finite union of sets from $\<U_{n,i}\>_{n,i \in \N}$, the machine computing $\<U_{n,i}\>_{n,i \in \N}$ can be easily modified to construct a machine witnessing
that $\<V_m\>_{m=1}^{\infty}$ is a $\SUBEXP$-space approximable sequence of
sets. From the above observations and Lemma \ref{lem:mpolylogsummation}, it follows that $\<V_m\>_{m=1}^{\infty}$ is a $\SUBEXP$-space Solovay test. Now,
let $x$ be a $\SUBEXP$-space random. $x$ is in at most finitely many $V_m$
and hence in at most finitely many $U_{n,i}$. Hence for some large
enough $N$ for all $n \geq N$, $i \geq 0$ and for all $m$ such that
$p(n+i) \leq m \leq p(n+i+1)$, we have
$|A^{f_m}_n(x)-A^{f_{p(n+i)}}_n(x)| <2^{-\log(n+i)^2}$. It follows
that for all $n \geq N$ and for all $m \geq p(n)$ that,
\begin{align*}
  |A^{f_m}_n(x)-g_n(x)| &= |A^{f_m}_n(x)-A^{f_{p(n)}}_n(x)| 
  \ \leq\  \sum\limits_{i=0}^{\infty} \frac{1}{2^{\log(n+i)^2}} 
  \ \leq \sum\limits_{i=0}^{\infty} \frac{1}{(n+i)^3}
  \ \leq\  \frac{\pi^2}{6n}.
\end{align*}
The last inequality follows from Lemma \ref{lem:summationlemma}. Therefore, $\lim\limits_{n \to \infty} \widetilde{A}^f_n (x) =
\lim\limits_{n \to \infty} g_n(x)$ on all $\SUBEXP$-space random $x$.
Hence, we have shown that $\lim\limits_{n \to \infty} \widetilde{A}^f_n = \int f d\mu$ on $\SUBEXP$-space randoms which completes the proof of the theorem.
\end{proof}

An important reason for investigating $\SUBEXP$-space randomness is that the $\SUBEXP$-space ergodic theorem has an exact converse unlike the $\PSPACE$ ergodic theorem which only seems to have a partial converse (Theorem \ref{thm:conversetopspaceergodictheorem}). Before proving the converse, we show a $\SUBEXP$-space analogue of the construction in Lemma \ref{lem:finitepspacetest}.

\begin{lemma}
\label{lem:finitesubexptest}
Let $\<U_n\>_{n=1}^{\infty}$ be a $\SUBEXP$-space test. Then there exists a sequences of sets $\<\widehat{S}_n\>_{n=1}^{\infty}$ such that for each $n \in \N$, $\widehat{S}_n \subseteq \Sigma^*$ satisfying the following conditions:
\begin{enumerate}
    \item $\mu([\widehat{S}_n]) \leq 2^{-\log(n)^2}$
	\item  $\cap_{m=1}^{\infty}\cup_{n=m}^{\infty} [\widehat{S}_n]  \supseteq \cap_{n=1}^{\infty} U_n$
	\item There exists $c \in \N$ such that for all $n$, $\sigma \in \widehat{S}_n $ implies $\lvert \sigma \rvert\leq 2^{\log(n)^c}$
	\item There exists a $\PSPACE$ machine $N$ such that
          $N(\sigma,1^n)=1$ if $\sigma \in \widehat{S}_n$ and 0
          otherwise.
\end{enumerate}
\end{lemma}

\begin{proof}
Let $\<S_n^k\>_{n,k}$ be the collection of approximating sets for $\<U_n\>_{n=1}^{\infty}$ as in Definition \ref{def:subexpopensets}. Define $\mathcal{U}=\cup_{n=1}^{\infty}\cup_{k=1}^{\infty}S_n^k$. Now, define
\begin{align*}
T_n= \bigcup\limits_{i=1}^{2(4n^2+1)} S_{i}^{2(4n^2+1)2n}	
\end{align*}
 Observe that,
 \begin{align*}
 [\mathcal{U}]\setminus [T_n] \subseteq \left( \bigcup\limits_{i=1}^{2(4n^2+1)} U_i - [S_{i}^{2(4n^2+1)2n}]  \right) \bigcup \left(\bigcup\limits_{i=2(4n^2+1)+1}^{\infty} U_i \right)
 \end{align*}
Hence,
\begin{align*}
\mu([\mathcal{U}]\setminus [T_n]) &\leq 	\sum\limits_{i=1}^{2(4n^2+1)} \frac{1}{2^{\log(2(4n^2+1)2n)^2}} + \sum\limits_{i=2(4n^2+1)+1}^{\infty} \frac{1}{2^{\log(i)^2}}\\
&< 	\sum\limits_{i=1}^{2(4n^2+1)} \frac{1}{2^{\log(2(4n^2+1)2n)^2}} + \sum\limits_{i=2(4n^2+1)}^{\infty} \frac{1}{2^{\log(i)^2}}\\
&\leq \frac{2(4n^2+1)}{2^{\log(2(4n^2+1)2n)^2}} + \frac{1}{2^{\log(2n)^2}}\\
&< \frac{2(4n^2+1)}{2^{\log(2(4n^2+1))^2 +\log(2n)^2}} + \frac{1}{2^{\log(2n)^2}}\\
&< \frac{1}{2^{\log(n)^2+1}} + \frac{1}{2^{\log(n)^2+1}}\\
&= \frac{1}{2^{\log(n)^2}}
\end{align*}
The third inequality follows from Lemma \ref{lem:mpolylogsummation}. From the definition of $T_n$, it follows that there is a $c \in \N$ such that the length of strings in $T_n$ is upper bounded by $2^{\log(n)^c}$. Now, if $\widehat{S}_n =\{\sigma \in T_{n+1}: \forall \alpha \sqsubseteq \sigma (\alpha \not\in T_n)\}$, we have $\mu([\widehat{S}_n])\leq \mu([\mathcal{U}]\setminus [T_n]) \leq 2^{-\log(n)^2}$. Conditions 2 and 3 can be readily verified to be true. We now construct a $\PSPACE$ machine $N$ satisfying the condition in 4. $N$ on input $(\sigma,1^n)$ does the following:
\begin{enumerate}
	\item For each $i\in[1,2(4(n+1)^2+1)]$ simulate $M(\sigma,1^i,1^{2(4(n+1)^2+1)2(n+1)})$. If all simulations result in $0$, output $0$.
	\item Else, for each $m\in [1,n]$ do the following:
	\begin{enumerate}
		\item For each $\alpha \sqsubseteq \sigma$ do the following:
		\begin{enumerate}
			\item For each $i\in[1,2(4m^2+1)]$ simulate $M(\alpha,1^i,1^{2(4m^2+1)2m})$. If any of these simulations result in a $1$ then, output $0$.
		\end{enumerate}
	\end{enumerate}
	\item Output $1$.
\end{enumerate} 
$N$ can be easily verified to be a $\PSPACE$ machine. Hence, our constructions satisfy all the desired conditions.
\end{proof}

\begin{theorem}
\label{thm:conversetosubexpergodictheorem}
	Given any $\SUBEXP$-space null $x$, there exists a $\SUBEXP$-space $L^1$-computable $f \in L^{1}(\Sigma^\infty,\mu)$ such that for any $\PSPACE$ simple measure preserving transformation, the following conditions are true:
	\begin{enumerate}
		\item For all $n\in \N$, $\lVert A_n^f -\int f d\mu \rVert_1=0$. Hence, $\int f d\mu$ is an $\PSPACE$-rapid $L^1$-limit point of $A^f_n$. 
		\item There exists a collection of simple functions $\{g_{n,i}\}$ such that for each $n$, $\<g_{n,i}\>_{i=1}^{\infty}$ is a $\SUBEXP$-space $L^1$-approximation of $A^f_n$ but $\lim\limits_{n \to \infty}\lim\limits_{i \to \infty} g_{n,i}(x)\neq\int f d\mu$.
	\end{enumerate}
\end{theorem}
\begin{proof}
	Let $\<V_n\>_{n=1}^\infty$ be any $\SUBEXP$-space test such that $x \in \cap_{n=1}^{\infty} V_n$. Now, from Lemma \ref{lem:finitesubexptest}, there exists a collection of sets $\<\widehat{S}_n\>_{n=1}^{\infty}$ such that $\cap_{m=1}^{\infty}\cup_{n=m}^{\infty} [\widehat{S}_n]  \supseteq \cap_{n=1}^{\infty} V_n$ and $\sigma \in \widehat{S}_n $ implies $\lvert \sigma \rvert\leq 2^{\log(n)^c}$ for some $c$. Let, 
	\begin{align*}
		U_n=\{\sigma:[\sigma] \in \widehat{S_i} \text{ for some } i \text{ such that } 2(2(n^2)^2+1) \leq i \leq 2(2((n+1)^2)^2+1) \}
	\end{align*}
Now, let $f_n=n\chi_{U_n}$. Using Lemma \ref{lem:mpolylogsummation}, 
\begin{align*}
	\mu(U_n)\leq \sum\limits_{i=2(2(n^2)^2+1)}^{2(2((n+1)^2)^2+1)} \frac{1}{2^{\log(i)^2}} \leq \frac{1}{2^{\log(n^2)^2}} \leq \frac{1}{2^{\log(n)^2+\log(n)^2}}
\end{align*}
it follows that
\begin{align*}
\lVert f_n \rVert_1 \leq 	\frac{n}{2^{\log(n)^2+\log(n)^2}} \leq \frac{1}{2^{\log(n)^2}}.
\end{align*}
Now, using the properties of $\<\widehat{S}_n\>_{n=1}^{\infty}$, it can be shown that $\<f_n\>_{n=1}^{\infty}$ is a $\SUBEXP$-space $L^1$-approximation of $f=0$.
	We construct a machine $M$ computing $\<f_n\>_{n=1}^{\infty}$. The other conditions are easily verified. Let $N$ be the machine from Lemma \ref{lem:finitesubexptest}. On input $(1^n,\sigma)$, $M$ does the following:
\begin{enumerate}
	\item If $\lvert \sigma \rvert < 2^{\lceil \log(2(2((n+1)^2)^2+1))\rceil^c}$ then, output $?$.
	\item Else, for each $i \in [2(2(n^2)^2+1),2(2((n+1)^2)^2+1)]$ do the following:
	\begin{enumerate}
		\item For each $\alpha \subseteq \sigma$, do the following:
		\begin{enumerate}
		\item If $N(1^i,\alpha)=1$ then, output $n$.
		\end{enumerate}
	\end{enumerate}
	\item Output $0$.
\end{enumerate}
	
	$M$ uses at most polynomial space and computes $\<f_n\>_{n=1}^{\infty}$. Now, define,
	\begin{align*}
	g_{n,i}=\frac{f_i +f_i \circ T + \dots + f_i \circ T^{n-1}}{n}	
	\end{align*}
	For any fixed $n \in \N$, since $T$ is a $\PSPACE$ simple transformation, as in Lemma \ref{lem:subexpsequencelemma} it can be shown that $\<g_{n,i}\>_{i=1}^{\infty}$ is a $\SUBEXP$-space $L^1$-approximation of $A^f_n$. We know that there exist infinitely many $m$ such that $x \in \widehat{S}_m$. For any such $m$, let $i$ be the unique number such that $2(2(i^2)^2+1) \leq m \leq 2(2((n+1)^2)^2+1)$. For this $i$, $f_i(x)=i$. This shows that there exist infinitely many $i$ such that $f_i(x)=i$. Since each $f_i$ is a non-negative function, it follows that for infinitely many $i$ with $g_{n,i} \geq i/n$. Hence, if $\lim\limits_{i \to \infty}g_{n,i}(x)$ exists, then it is equal to $\infty$. It may be the case that $\lim\limits_{i \to \infty}g_{n,i}(x)$ does not exist. In either case, $\lim\limits_{n \to \infty}\lim\limits_{i \to \infty}g_{n,i}(x)$ cannot be equal to $\int f d\mu=0$. Hence, our construction satisfies all the desired conditions.

\end{proof}

\section{Martingale characterization of $\PSPACE$ randomness}

The study of resource bounded randomness was initiated by Lutz(\cite{Lutz1992},\cite{Lutz1998}) using resource bounded martingales. Huang and Stull introduced weak $\PSPACE$ randomness (definition \ref{def:pspacetest}) in \cite{huangstull} as a resource bounded analogue of randomness defined in terms of c.e tests. In this section we give a characterization of $\PSPACE$ randoms in terms of $\PSPACE$ martingales, demarking on the difference between Lutz's notion of $\PSPACE$ randomness and our notion of $\PSPACE$ randomness. Our result and its proof are significantly different from those given in \cite{sureson2017}. We demonstrate this in the setting of $(\Sigma^\infty, \mathcal{B}(\Sigma^\infty),\mu)$ where $\mu$ is the Bernoulli measure $\mu([\sigma])=\frac{1}{2^{|\sigma|}}$. Now we define $\PSPACE$ computable martingales. 
\begin{definition}[$\PSPACE$ computable martingales]
	A function $D : \Sigma^* \to [0,\infty)$ is a $\PSPACE$ computable martingale if for each $\sigma \in \Sigma^*$,
	\begin{align*}
	D(\sigma)=\frac{D(\sigma 0)+D(\sigma 1)}{2}
	\end{align*}
	and there exists a $\PSPACE$ machine $M$ such that for each $\sigma \in \Sigma^*$ and $n \in \N$, $M(\sigma,1^n)\in \Q$ such that,
	\begin{align*}
	\left\lvert D(\sigma) - M(\sigma,1^n) \right\rvert \leq \frac{1}{2^n}
	\end{align*}
\end{definition}

Now, we give a characterization of $\PSPACE$ randomness in terms of $\PSPACE$ computable martingales.

\begin{theorem}
\label{thm:pspacemartingaletheorem}
$x \in \Sigma^\infty$ is $\PSPACE$ null if and only if there exists a $\PSPACE$ computable martingale $D$ and $k \in \N$ such that there exist infinite many $n \in \N$ satisfying
\begin{align}
\label{eq:martingalewincondition}
D(x \upharpoonleft n) \geq 2^{\lfloor n^{\frac{1}{k}}\rfloor / 2}	
\end{align}
\end{theorem}
$x \in \Sigma^\infty$ is \textit{strong $\PSPACE$ null} (Lutz \cite{Lutz1992}\cite{Lutz1998}) if there exists a $\PSPACE$ computable martingale such that $\limsup_{n \to \infty}D(x \upharpoonleft n)=\infty$. Hence, if $x$ is $\PSPACE$ null then $x$ is \textit{strong $\PSPACE$ null}. Stull \cite{stull2017algorithmic} proved the existence of a \textit{strong $\PSPACE$ null} which is $\PSPACE$ random. The proof below is a careful adaptation of Theorem 7.3.3 from \cite{Nies2009} into the $\PSPACE$ setting using the construction in Lemma \ref{lem:finitepspacetest}.
\begin{proof}[Proof of Theorem \ref{thm:pspacemartingaletheorem}]
We show the forward implication first. Let $\<U_n\>_{n=1}^{\infty}$ be a $\PSPACE$ test witnessing the fact that $x$ is $\PSPACE$ null.  Now, from Lemma \ref{lem:finitepspacetest}, there exists a collection of sets $\<\widehat{S}_n\>_{n=1}^{\infty}$ such that $\cap_{m=1}^{\infty}\cup_{n=m}^{\infty} [\widehat{S}_n]  \supseteq \cap_{n=1}^{\infty} U_n$ and $\sigma \in \widehat{S}_n $ implies $\lvert \sigma \rvert < n^k$ for some $k$. Let $\mathcal{U}=\bigcup_{n=1}^{\infty} [\widehat{S}_n]$. Hence,


\begin{align*}
  \sum\limits_{\substack{\sigma \in \mathcal{U} \\ \lvert\sigma\rvert\geq n^k}}
  \frac{1}{2^{\lvert\sigma \rvert}}
  \ \leq\ 
  \sum_{n' > n}
  \sum\limits_{\sigma \in \widehat{S}_{n'}}
 \frac{1}{2^{\lvert\sigma\rvert}}
 \ \leq\  \frac{1}{2^n}.
\end{align*}
Hence,
\begin{align*}
\sum\limits_{\substack{\sigma \in \mathcal{U} \\ \lvert\sigma\rvert\geq n}} \frac{1}{2^{\lvert\sigma \rvert}} \leq \frac{1}{2^{\lfloor n^{\frac{1}{k}}\rfloor}} = \frac{1}{2^{2\lfloor n^{\frac{1}{k}}\rfloor/2}} 
\end{align*}
Let $f(n)=\lfloor n^{\frac{1}{k}}\rfloor/2$. Hence,
\begin{align*}
\sum\limits_{\substack{\sigma \in \mathcal{U} \\ \lvert\sigma\rvert\geq n}} \frac{1}{2^{\lvert\sigma \rvert}} \leq \frac{1}{2^{2 f(n)}} 
\end{align*}
For every $r
\in \N$, define
\begin{align*}
s_r = \sum\limits_{\substack{\sigma \in \mathcal{U} \\ f(\lvert \sigma \rvert)\geq r}} 2^{f(\lvert \sigma\rvert)-\lvert\sigma\rvert}	
\end{align*}
It can be easily verified that $s_r \leq 2^{-r}$. 

Now, we define a martingale $D$ such that $D$ wins on $x$ in the sense
of \ref{eq:martingalewincondition}. For each $i \in \N$, let $G_i =
\{\sigma \in \mathcal{U}:f(\lvert\sigma \rvert)\geq i\}$. It can be
readily verified that $\mu(G_i)\leq 2^{-i}$ for all $i>0$. For each $i
\in \N$ define
\begin{align*}
E_{i}(\sigma) = \sum\limits_{\substack{\alpha \in G_i \\ \alpha \sqsubset \sigma}} 2^{f(\lvert\alpha\rvert)}	
\end{align*}
\begin{align*}
F_{i}(\sigma) = 2^{\lvert \sigma \rvert}\sum\limits_{\substack{\alpha \in G_i \\ \alpha \sqsupseteq \sigma}} 2^{f(\lvert\alpha\rvert)-\lvert \alpha \rvert}	
\end{align*}
As in the proof of \cite{Nies2009} Lemma 7.3.4, it can be verified that $D_i = E_i + F_i	$ is a martingale. Now define $D=\sum_{i=1}^{\infty}D_i$. It can be easily seen that $D$ is a martingale.

We now show that $D$ is $\PSPACE$ computable. Let $\sigma$ be a string on which $D(\sigma)$ needs to be approximated with error at most $2^{-m}$. Observe that,
\begin{align*}
E_i (\sigma) &= 	\sum\limits_{\substack{\alpha \in G_i \\ \alpha \sqsubset \sigma}}2^{\lvert \alpha \rvert} 2^{f(\lvert\alpha\rvert)-\lvert\alpha \rvert}\\
&\leq 2^{\lvert \sigma \rvert} \sum\limits_{\substack{\alpha \in G_i \\ \alpha \sqsubset \sigma}}2^{f(\lvert\alpha\rvert)-\lvert\alpha \rvert}
\end{align*}
Now,
\begin{align*}
D_i (\sigma) &= E_i (\sigma) + F_i (\sigma) \\
&\leq 2^{\lvert\sigma \rvert} \sum\limits_{\alpha \in G_i}2^{f(\lvert\alpha\rvert)-\lvert\alpha \rvert} \\
&\leq 2^{\lvert\sigma \rvert} \sum\limits_{\substack{\alpha \in \mathcal{U}\\ f(\lvert \alpha \rvert)\geq i}}2^{f(\lvert\alpha\rvert)-\lvert\alpha \rvert}\\
&\leq 2^{\lvert \sigma \rvert}s_i \\
&\leq 2^{\lvert \sigma \rvert-i}
\end{align*}
Hence, if we can approximate $\sum_{i=1}^{\lvert \sigma \rvert + m +1}D_i$ with error at most $2^{-(m+1)}$, we can get a $2^{-m}$ error approximation for $D$.

Let $N$ be a machine computing $\<\widehat{S}_n\>_{n=1}^{\infty}$ as in Lemma \ref{lem:finitepspacetest}. For the given $\sigma$, we can sum $2^{f(\lvert\alpha \rvert)}$ for all $\alpha \sqsubset \sigma$ and $\alpha \in G_i$ in $\PSPACE$. For each $i$, $E_i$ can be computed in the following way:
\begin{enumerate}
	\item Let $\mathrm{Sum}=0$.
	\item For each $\alpha \sqsubset \sigma$ do the following: 
	\begin{enumerate}
	\item Check if there exists $1 \leq m \leq \lvert \alpha\rvert$ such that $N(1^m,\alpha)=1$. If yes, check if $f(\lvert \alpha\rvert)\geq i$. If this is true, then add $f(\lvert \alpha\rvert)$ to $\mathrm{Sum}$.
	\end{enumerate}
	\item Output $\mathrm{Sum}$.
\end{enumerate}
$\alpha \not \in \widehat{S}_n$ for $n \geq \lvert \alpha\rvert+1$ since for any such $n$, $\mu(\widehat{S}_n) < 2^{-\lvert \alpha \rvert}$. Hence, the procedure above exactly computes $E_i$. The procedure can be easily seen to be a polynomial space operation. Hence, $\sum_{i=1}^{\lvert \sigma \rvert + m +1}E_i$ can be computed without error in polynomial space. 

Observe that for all $m \in \N$,
\begin{align*}
2^{\lvert\sigma \rvert}\sum\limits_{\alpha:f(\lvert\alpha \rvert)\geq m+\lvert\sigma\rvert+1} 2^{f(\lvert\alpha \rvert)-\lvert \alpha \rvert} \leq 2^{\lvert\sigma \rvert}\sum\limits_{r= m+\sigma+1}^{\infty} s_r \leq 2^{-m}.
\end{align*}
So, we can ignore all but finitely many terms in the summation defining $F_i$ as in the case of $D_i$. For any $\alpha$, if $f(\lvert\alpha \rvert)=\lfloor \lvert\alpha\rvert^{\frac{1}{k}}\rfloor/2< m+\lvert\sigma\rvert+1$ then its length $\lvert \alpha \rvert$ can be upper bounded by a polynomial in $m+\lvert\sigma\rvert$, say $(m+\lvert\sigma\rvert)^d$. Hence, the terms that are left in the summation defining $F_i$ can be computed as in the computation of $E_i$ by going over all $\alpha \sqsupseteq \sigma$ with $\lvert \alpha \rvert \leq (m+\lvert\sigma\rvert)^d$ in polynomial space. It follows that $D=F+E$ is a $\PSPACE$ computable martingale.

Now, we show that condition \ref{eq:martingalewincondition} is satisfied by $D$ for infinitely many $n \in \N$. We know that for infinitely many $n \in \N$, $x \upharpoonleft n \in \mathcal{U}$. For any such $n$, we have $x \upharpoonleft n \in G_{f(n)}$. Observe that,
\begin{align*}
D(x \upharpoonleft n) \geq D_{f(n)}(x \upharpoonleft n)	 \geq F_{f(n)}(x \upharpoonleft n) \geq 2^{n}.2^{f(n)-n}=2^{f(n)}
\end{align*}
Hence, \ref{eq:martingalewincondition} holds for infinitely many $n \in \N$ as desired.

	Now we prove the converse. Let $U_i=\{\sigma \in \Sigma^* : M(\sigma,1)\geq 2^{\lfloor \lvert \sigma \rvert^{\frac{1}{k}}\rfloor/2}-1 \geq 2^{i+1}\}$. Since $M(\sigma,1)\geq 2^{i+1}$ implies $d(\sigma)>2^i$, using the Kolmogorov's inequality (see \cite{Nies2009} Proposition 7.1.9) we get that $\mu(U_i)\leq 2^{-i}$. For any $m \in \N$, $\lvert \sigma \rvert \geq (2(m+2))^k$ implies $2^{\lfloor \lvert \sigma \rvert^{\frac{1}{k}}\rfloor/2} \geq 2^{m+2}$. Hence, for any $\sigma \in U_i$, $\lvert \sigma \rvert \geq (2(m+2))^k$ implies that $d(\sigma)\geq 2^{m+2}-2 > 2^m$. Now, using the Kolmogorov's inequality again we get that $\mu\left(\left\{\sigma \in U_i : |\sigma|\geq (2(m+2))^k\right\}\right)\leq 2^{-m}$. From the two previous observations, it follows that $\<U_i\>_{i=1}^{\infty}$ is a $\PSPACE$ sequence of open sets. Since, there are infinitely many $n \in \N$ such that $D(x \upharpoonleft n) \geq 2^{\lfloor  n ^{\frac{1}{k}}\rfloor/2}$, it follows that there are infinitely many $n$ such that $M(x \upharpoonleft n,1)\geq 2^{\lfloor  n ^{\frac{1}{k}}\rfloor/2}-1$. Since, $\<U_i\>_{i=1}^{\infty}$ is a $\PSPACE$ test, it follows that $x$ is $\PSPACE$ null.
\end{proof}

\section{A $\PSPACE$ random which is $\EXP$ non-random}
\label{sec:pexpseparation}

It was shown in Stull's thesis \cite{stull2017algorithmic} that weak polynomial time randomness is strictly weaker than polynomial time randomness. Using a similar approach, we give below an explicit construction of a $\PSPACE$ random which is $\EXP$ non-random. Let $x$ be any Martin-L\"of random. We define the infinite sequence $y$ as follows,
\begin{align*}
y[n] = \begin{cases}
	0, &\text{if }n=2^m \text{ for some }m\in \N \\
	x[n], &\text{otherwise.}
\end{cases}
\end{align*}
It is easy to construct a $\PSPACE$ computable martingale $D : \Sigma^* \to [0,\infty)$ such that $\limsup_{n \to \infty}D(y \upharpoonleft n)=\infty$. Define,
\begin{align*}
D(\sigma) = \begin{cases}
	2^{\lfloor \log(\lvert \sigma \rvert) \rfloor}, &\text{if }\sigma[2^i]=0 \text{ for all }1\leq i \leq \lfloor \log(\lvert \sigma \rvert) \rfloor \\
	0, &\text{otherwise.}
\end{cases}
\end{align*}
It is easy to verify that $D$ is a $\PSPACE$ computable martingale and $\limsup_{n \to \infty}D(y \upharpoonleft n)=\lim_{n \to \infty}2^{\lfloor \log(n) \rfloor}=\infty$. But this does not imply that $y$ is $\PSPACE$ null since $D$ does not \textit{win} on $y$ \textit{fast enough} as in condition \ref{eq:martingalewincondition}. In fact, we now show that no $\PSPACE$ computable martingale satisfying \ref{eq:martingalewincondition} on $y$ exists since $y$ is $\PSPACE$ random.
 \begin{lemma}
\label{lem:pspacerandomexpspacenull}
$y$ is $\PSPACE$ random and $\EXP$ null.	
\end{lemma}
\begin{proof}[Proof]
	We first construct an $\EXP$ test $\<U_n\>_{n=1}^{\infty}$ such that $y \in \cap_{n=1}^{\infty} U_n$. For any $n$, let $U_n$ be the collection of all $2^n$ length strings $\sigma$ such that $\sigma[i]=0$ for all $i \in \{2^1,2^2,\dots 2^n\}$. Clearly, $\mu(U_n)\leq 2^{-n}$ and $y \in \cap_{n=1}^{\infty} U_n$. It is straightforward to verify that $\<U_n\>_{n=1}^{\infty}$ is an $\EXP$ test.
	
	Now, we show that $y$ is $\PSPACE$ random. Assume that there exists $\<\widehat{S}_n\>_{n=1}^{\infty}$ satisfying the properties in Lemma \ref{lem:finitepspacetest} such that $y \in \cap_{m=1}^{\infty}\cup_{n=m}^{\infty} [\widehat{S}_n]$. Hence, there exists $c \in \N$ such that for any $n$, $\sigma \in \widehat{S}_{n}$ implies $\lvert \sigma \rvert \leq n^c$. Let $G_n$ be the set of all $\alpha \in \Sigma^{n^c}$ such that there exists some $\sigma \in \widehat{S}_{n}$ with $\alpha[i]=\sigma[i]$ for all $i \not\in \{2^j:1\leq j \leq \lfloor \log(n^c) \rfloor\}$. If $y \in \widehat{S}_{n}$ then, by construction it follows that $x \in G_n$. Since $\mu(\widehat{S}_{n})\leq 2^{-n}$, we get $\mu(G_n) \leq 2^{\lfloor \log(n^c) \rfloor}2^{-n}$. Since for all large enough $n$, $n-\lfloor \log(n^c) \rfloor \geq n/2$, we have $\mu(G_n) \leq \sqrt{2}^{-n}$ for large enough $n$. Hence, a Martin-L\"of test  capturing $x$ can be constructed from $\<G_{n}\>_{n=1}^{\infty}$. Since this is a contradiction, our assumption that $\<\widehat{S}_{n}\>_{n=1}^{\infty}$ exists must be wrong. Hence, it follows that $y$ is $\PSPACE$ random.
\end{proof}

\bibliographystyle{plainurl}
\bibliography{fair001,main,random}

\end{document}